\documentclass[a4paper,USenglish,cleveref,autoref,thm-restate]{article}
\usepackage{doi}
\usepackage{fullpage}
\usepackage{amsmath}
\usepackage{amssymb}
\usepackage{amsthm}
\usepackage{xspace}
\usepackage{complexity}
\usepackage{authblk}
\usepackage{tikz}

\usetikzlibrary{shapes,backgrounds,plothandlers,plotmarks,calc,arrows,fadings,patterns,decorations.pathmorphing}

\pgfdeclarepatternformonly[\LineSpace,\LineWidth]{my horizontal lines}%
    {\pgfpointorigin}{\pgfqpoint{100pt}{1pt}}{\pgfqpoint{100pt}{\LineSpace}}%
    {
        \pgfsetcolor{\tikz@pattern@color}
        \pgfsetlinewidth{\LineWidth}
        \pgfpathmoveto{\pgfqpoint{0pt}{0.5pt}}
        \pgfpathlineto{\pgfqpoint{100pt}{0.5pt}}
        \pgfusepath{stroke}
    }
 
\pgfdeclarepatternformonly[\LineSpace,\LineWidth]{my vertical lines}%
    {\pgfpointorigin}{\pgfqpoint{1pt}{100pt}}{\pgfqpoint{\LineSpace}{100pt}}%
    {
        \pgfsetcolor{\tikz@pattern@color}
        \pgfsetlinewidth{\LineWidth}
        \pgfpathmoveto{\pgfqpoint{0.5pt}{0pt}}
        \pgfpathlineto{\pgfqpoint{0.5pt}{100pt}}
        \pgfusepath{stroke}
    } 
 
\pgfdeclarepatternformonly[\LineSpace,\LineWidth]{my grid}%
    {\pgfqpoint{-1pt}{-1pt}}{\pgfqpoint{\LineSpace}{\LineSpace}}
    {\pgfqpoint{\LineSpace}{\LineSpace}}%
    {
        \pgfsetcolor{\tikz@pattern@color}
        \pgfsetlinewidth{\LineWidth}
        \pgfpathmoveto{\pgfqpoint{0pt}{0pt}}
        \pgfpathlineto{\pgfqpoint{0pt}{\LineSpace + 0.1pt}}
        \pgfpathmoveto{\pgfqpoint{0pt}{0pt}}
        \pgfpathlineto{\pgfqpoint{\LineSpace + 0.1pt}{0pt}}
        \pgfusepath{stroke}
    }
 
\pgfdeclarepatternformonly[\LineSpace,\LineWidth]{my north east lines}
    {\pgfqpoint{-\LineWidth}{-\LineWidth}}{\pgfqpoint{\LineSpace}{\LineSpace}}
    {\pgfqpoint{\LineSpace}{\LineSpace}}%
    {
        \pgfsetcolor{\tikz@pattern@color}
        \pgfsetlinewidth{\LineWidth}
        \pgfpathmoveto{\pgfqpoint{-\LineWidth}{-\LineWidth}}
        \pgfpathlineto{\pgfqpoint{\LineSpace + 0.1pt}{\LineSpace + 0.1pt}}
        \pgfusepath{stroke}
    }
 
\pgfdeclarepatternformonly[\LineSpace,\LineWidth]{my north west lines}
    {\pgfqpoint{-\LineWidth}{-\LineWidth}}{\pgfqpoint{\LineSpace}{\LineSpace}}
    {\pgfqpoint{\LineSpace}{\LineSpace}}%
    {
        \pgfsetcolor{\tikz@pattern@color}
        \pgfsetlinewidth{\LineWidth}
        \pgfpathmoveto{\pgfqpoint{-\LineWidth}{\LineSpace}}
        \pgfpathlineto{\pgfqpoint{\LineSpace + 0.1pt}{-\LineWidth}}
        \pgfusepath{stroke}
    }
 
\pgfdeclarepatternformonly[\LineSpace,\LineWidth]{my crosshatch}%
    {\pgfqpoint{-1pt}{-1pt}}{\pgfqpoint{\LineSpace}{\LineSpace}}
    {\pgfqpoint{\LineSpace}{\LineSpace}}%
    {
        \pgfsetcolor{\tikz@pattern@color}
        \pgfsetlinewidth{\LineWidth}
        \pgfpathmoveto{\pgfqpoint{\LineSpace + 0.1pt}{0pt}}
        \pgfpathlineto{\pgfqpoint{0pt}{\LineSpace + 0.1pt}}
        \pgfpathmoveto{\pgfqpoint{0pt}{0pt}}
        \pgfpathlineto{\pgfqpoint{\LineSpace + 0.1pt}{\LineSpace + 0.1pt}}
        \pgfusepath{stroke}
    }
\pgfdeclarepatternformonly[\LineSpace,\LineWidth]{my dots}%
    {\pgfpointorigin}
    {\pgfpoint{\LineSpace}{\LineSpace*0.866025*2}}
    {\pgfpoint{\LineSpace}{\LineSpace*0.866025*2}}%
    {
        \pgfsetcolor{\tikz@pattern@color}
        \pgfpathcircle{\pgfpoint{\LineSpace*0.25}{\LineSpace*0.866025*0.5}}{\LineWidth}
        \pgfpathcircle{\pgfpoint{\LineSpace*0.75}{\LineSpace*0.866025*1.5}}{\LineWidth}        
        \pgfusepath{fill}
    }
\makeatother
 
\newdimen\LineSpace
\newdimen\LineWidth
\tikzset{
    pattern space/.code={\LineSpace=#1},
    pattern space=3pt,
    pattern width/.code={\LineWidth=#1},
    pattern width=.4pt
}

\definecolor{cdred}{RGB}{200,20,20}
\definecolor{cdgreen}{RGB}{220,255,200}	
\definecolor{cdblue}{RGB}{130,145,255}	
\tikzset{ %
	dred/.style={pattern=my north west lines,pattern width=.6pt,pattern space=2.3pt,pattern color=red},
	dgreen/.style={pattern=my grid,pattern width=0.8pt,pattern space=2pt,pattern color=green!80!black},
	dblue/.style={pattern=my dots,pattern width=0.7pt,pattern space=2.3pt,pattern color=blue},	
	H/.style={circle,fill=gray!20,draw=gray!20,line width=2pt,inner sep=0pt,minimum size=15pt},
	He/.style={draw=gray!20,line width=12pt},
	He2/.style={draw=gray!20,line width=6pt},
	G/.style={circle,fill=black,inner sep=0pt,minimum size=3pt},
	G1/.style={G,minimum size=5pt},
	Ge/.style={draw=black},
	v/.style={circle,draw=black!75,inner sep=0pt,minimum size=8pt},
	v1/.style={v,line width=1.5pt},
	Gv/.style={inner sep=0pt},
	Gve/.style={Ge,dashed},
}

\usepackage{todonotes}
\usepackage{extarrows}
\usepackage[linesnumbered,ruled]{algorithm2e}
\usepackage{booktabs}
\usepackage{authblk}
\usepackage{subcaption}

\newtheorem{theorem}{Theorem}
\newtheorem{lemma}[theorem]{Lemma}
\newtheorem{definition}{Definition}

\newtheorem{corollary}[theorem]{Corollary}
\newtheorem{remark}[theorem]{Remark}
\theoremstyle{definition}
\theoremstyle{remark}
\newtheorem{claim}{Claim}
\newtheorem{myclaim}{Claim}

\newcommand{\yes}{\ensuremath{\mathsf{YES}}\xspace}

\newcommand*{\recol}[1]{\textsc{Recol(\ensuremath{#1})}\xspace}
\newcommand*{\recolun}[1]{\textsc{Recol$_1$(\ensuremath{#1})}\xspace}
\newcommand*{\recollong}[1]{$#1$-Recoloring\xspace}
\newcommand*{\csp}[1]{\textsc{CSP(\ensuremath{#1})}\xspace}

\newcommand*{\undir}[1]{\ensuremath{\bar{#1}}\xspace}
\newcommand*{\loops}[1]{\ensuremath{{{#1}^\circ}}\xspace}
\newcommand*{\Hom}{\ensuremath{\operatorname{Hom}}\xspace}
\newcommand*{\IN}{\textsc{in}\xspace}
\newcommand*{\OUT}{\textsc{out}\xspace}
\newcommand*{\SYM}{\textsc{sym}\xspace}
\newcommand*{\pop}{push-or-pull\xspace}

\newcommand{\dtwosqgraph}{
	\begin{scope}
		\node[H](ha0) at (0:1) {};
		\node[H](ha1) at (90:1) {};
		\node[H](ha2) at (180:1) {};
		\node[H](ha3) at (-90:1) {};
		\draw[He]
			(ha0.center)--(ha1.center)
			(ha1.center)--(ha2.center)
			(ha2.center)--(ha3.center)
			(ha3.center)--(ha0.center);
	\end{scope}
	\begin{scope}[shift={(3,0)}]
		\node[H](hb0) at (180:1) {};
		\node[H](hb1) at (90:1) {};
		\node[H](hb2) at (0:1) {};
		\node[H](hb3) at (-90:1) {};
		\draw[He]
			(hb0.center)--(hb1.center)
			(hb1.center)--(hb2.center)
			(hb2.center)--(hb3.center)
			(hb3.center)--(hb0.center);
	\end{scope}
	\draw[He] (ha0.center)--(hb0.center);
}

\newcommand{\dhexasqgraph}{
	\begin{scope}
		\node[H](ha0) at (0:1) {};
		\node[H](ha1) at (60:1) {};
		\node[H](ha2) at (120:1){};
		\node[H](ha3) at (180:1) {};
		\node[H](ha4) at (-120:1) {};
		\node[H](ha5) at (-60:1) {};
		\draw[He]
			(ha0.center)--(ha1.center)
			(ha1.center)--(ha2.center)
			(ha2.center)--(ha3.center)
			(ha3.center)--(ha4.center)
			(ha4.center)--(ha5.center)
			(ha5.center)--(ha0.center);
	\end{scope}
	\begin{scope}[shift={(3,0)}]
		\node[H](hb0) at (180:1) {};
		\node[H](hb1) at (90:1) {};
		\node[H](hb2) at (0:1) {};
		\node[H](hb3) at (-90:1) {};
		\draw[He]
			(hb0.center)--(hb1.center)
			(hb1.center)--(hb2.center)
			(hb2.center)--(hb3.center)
			(hb3.center)--(hb0.center);
	\end{scope}
	\draw[He] (ha0.center)--(hb0.center);
}

\newcommand{\dbubblegraph}{
	\begin{scope}
		\node[H](ha0) at (0:1) {};
		\node[H](ha1) at (60:1) {};
		\node[H](ha2) at (120:1) {};
		\node[H](ha3) at (180:1) {};
		\node[H](ha4) at (-120:1) {};
		\node[H](ha5) at (-60:1) {};
		\draw[He]
			(ha0.center)--(ha1.center)
			(ha1.center)--(ha2.center)
			(ha2.center)--(ha3.center)
			(ha3.center)--(ha4.center)
			(ha4.center)--(ha5.center)
			(ha5.center)--(ha0.center);
		\node[H](hb0) at (0:2) {};
		\node[H](hb1) at (60:2) {};
		\node[H](hb2) at (120:2) {};
		\node[H](hb3) at (180:2) {};
		\node[H](hb4) at (-120:2) {};
		\node[H](hb5) at (-60:2) {};
		\draw[He]
			(hb0.center)--(hb1.center)
			(hb1.center)--(hb2.center)
			(hb2.center)--(hb3.center)
			(hb3.center)--(hb4.center)
			(hb4.center)--(hb5.center)
			(hb5.center)--(hb0.center);
		\draw[He]
			(ha0.center)--(hb0.center)
			(ha2.center)--(hb2.center)
			(ha4.center)--(hb4.center);
	\end{scope}
			
}

\newcommand{\arrowIn}{
\tikz \draw[He2,-stealth,gray!20] (-1pt,0) -- (1pt,0);
}

\newcommand{\arrowInG}{
\tikz \draw[Ge,-stealth] (-1pt,0) -- (1pt,0);
}

\newcommand{\dstargraph}{
    \begin{scope}
        \node[H](ha0) at (0:2){};
        \node[H](ha1) at (72:2){};
        \node[H](ha2) at (144:2){};
        \node[H](ha3) at (-144:2){};
        \node[H](ha4) at (-72:2){};
        \draw[He2] (ha0.center)--(ha1.center) node[sloped,pos=0.65,allow upside down]{\arrowIn}; ; 
        \draw[He2] (ha1.center)--(ha2.center) node[sloped,pos=0.65,allow upside down]{\arrowIn}; ; 
        \draw[He2] (ha2.center)--(ha3.center) node[sloped,pos=0.65,allow upside down]{\arrowIn}; ; 
        \draw[He2] (ha3.center)--(ha4.center) node[sloped,pos=0.65,allow upside down]{\arrowIn}; ; 
        \draw[He2] (ha4.center)--(ha0.center) node[sloped,pos=0.65,allow upside down]{\arrowIn}; ; 
        
        \draw[He2] (ha0.center)--(ha2.center) node[sloped,pos=0.6,allow upside down]{\arrowIn}; ; 
        \draw[He2] (ha1.center)--(ha3.center) node[sloped,pos=0.6,allow upside down]{\arrowIn}; ; 
        \draw[He2] (ha2.center)--(ha4.center) node[sloped,pos=0.6,allow upside down]{\arrowIn}; ; 
        \draw[He2] (ha3.center)--(ha0.center) node[sloped,pos=0.6,allow upside down]{\arrowIn}; ; 
        \draw[He2] (ha4.center)--(ha1.center) node[sloped,pos=0.6,allow upside down]{\arrowIn}; ; 
    \end{scope}
}

\newcommand{\dstardiamgraph}{
    \dstargraph;
    \begin{scope}[shift={(6,0)}]
    
        \node[H](hb0) at (-180:2){};
        \node[H](hb1) at (90:2){};
        \node[H](hb2) at (0:2){};
        \node[H](hb3) at (-90:2){};
        \draw[He2] (hb0.center)--(hb1.center) node[sloped,pos=0.65,allow upside down]{\arrowIn}; ; 
        \draw[He2] (hb1.center)--(hb2.center) node[sloped,pos=0.65,allow upside down]{\arrowIn}; ; 
        \draw[He2] (hb2.center)--(hb3.center) node[sloped,pos=0.65,allow upside down]{\arrowIn}; ; 
        \draw[He2] (hb3.center)--(hb0.center) node[sloped,pos=0.65,allow upside down]{\arrowIn}; ; 
        \draw[He2] (hb3.center)--(hb1.center) node[sloped,pos=0.65,allow upside down]{\arrowIn}; ; 
        \draw[He2] (ha0.center)--(hb0.center) node[sloped,pos=0.65,allow upside down]{\arrowIn}; ; 
    \end{scope}
}

\title{Reconfiguration of Digraph Homomorphisms}
\author[1]{  Benjamin Lévêque }
\author[1]{  Moritz Mühlenthaler }
\author[1]{  Thomas Suzan }
\affil[1]{Université Grenoble Alpes, Grenoble INP, G-SCOP, France}

\begin{document}

\maketitle

\begin{abstract}
  For a fixed graph $H$, the $H$-Recoloring problem asks whether for two given homomorphisms from a graph $G$ to $H$, we can transform one into the other by changing the image of a single vertex of $G$ in each step and maintaining a homomorphism from $G$ to $H$ throughout. We extend an algorithm of Wrochna for $H$-Recoloring where $H$ is a square-free loopless undirected graph to the more general setting of directed graphs.
  We obtain a polynomial-time algorithm for $H$-Recoloring in this setting whenever  $H$ is a loopless digraph that does not contain a 4-cycle of algebraic girth zero and whenever $H$ is a reflexive digraph that contains neither a $3$-cycle of algebraic girth $1$ nor a $4$-cycle of algebraic girth zero.
\end{abstract}

\section{Introduction}
\label{sec:intro}

Reconfiguration problems have been introduced formally by Ito et al.
in~\cite{Ito:11} and their complexity has been studied systematically since.
Applications can be found in statistical physics, combinatorial games,
and uniform sampling of objects such as colorings and matchings.
The general setting is the following: Given two feasible solutions of an
instance of a combinatorial problem, the goal is to decide whether one can be
transformed into the other in a step-by-step manner, visiting only feasible
configurations during the transformation.  Related questions of interest are whether any two feasible
solutions admit a transformation, and if there is a transformation of at most a
certain length between two given solutions. We refer the reader to the surveys
of Nishimura~\cite{Nishimura:18} and van den Heuvel~\cite{Heuvel:13} for a
discussion of results and applications in this area.

A digraph homomorphism maps the
vertex set of a digraph $G$ to the vertex set  of a digraph $H$ such that each arc of $G$ is mapped to an arc of $H$. 
The classical digraph homomorphism problem \csp{H} 
asks whether
there is a digraph homomorphism from a given digraph $G$ to a fixed ``template'' digraph~$H$. One of the motivations for studying \csp{H} is that it is polynomially
equivalent to the seemingly richer constraint satifaction problem \csp{\mathcal{H}}, where the template $\mathcal{H}$ can
be any fixed finite relational structure~\cite{FV:93}. The complexity of the
\csp{H} is well understood in the sense that  \csp{H} is known to be either
polynomial-time solvable or \NP-complete, a result that has been proved recently
by Bulatov~\cite{Bulatov:17} and by Zhuk~\cite{Zhuk:20}, settling in the
affirmative a long-standing conjecture by Feder and Vardi~\cite{FV:93}.
Motivated by these recent developments we study the complexity of the natural reconfiguration
variant \recol{H} (informally called \emph{\recollong{H}})  
associated with \csp{H}, which is the following question: Given two digraph
homomorphisms $\alpha$ and $\beta$ from a $G$ to $H$, is there a step-by-step transformation between
$\alpha$ and $\beta$ that changes the image of one vertex of $G$ at a time
and maintains a digraph homomorphism from $G$ to $H$ throughout?

Complexity results are known for several special cases of \recollong{H}, most
notably for deciding whether there is a step-by-step transformation between two
homomorphisms from a given undirected graph to a fixed undirected graph $H$.
Notice that for our purposes we may interpret an undirected graph to be the
digraph obtained by replacing each undirected edge by two directed edges of
opposite orientation.  Despite several
positive~\cite{Lee:21,Brewster:16,Cereceda:11,Wrochna:20} and negative~\cite{BC:09,Brewster:16,Lee:20} results in this setting, a complete
classification of the complexity of \recollong{H} for undirected graphs is not known. 
In the more general context of digraphs, we are aware of only two results
for \recollong{H}, which consider the case where $H$ is a transitive tournament~\cite{Dochtermann:21} and where $H$ is some orientation of a reflexive digraph cycle~\cite{Brewster:17} (\emph{reflexive} means that there is a loop on each vertex). We extend the topological approach of
Wrochna for \recollong{H} for undirected graphs to digraphs and obtain the
following results, which generalize the previous algorithmic results
in~\cite{Lee:21,Brewster:17,Cereceda:11,Wrochna:20}. The \emph{algebraic girth} of the
orientation of a cycle is the absolute value of the number of forward arcs minus
the number of backward arcs.  A \emph{4-cycle of algebraic girth 0} is either
one of the two graphs shown in Figure~\ref{fig:zerogirth}. 

\begin{theorem}
	Let $H$ be a loopless digraph that contains no 4-cycle of algebraic girth 0 as subgraph. Then \recollong{H} admits a polynomial-time algorithm.
	\label{thm:c4free}
\end{theorem}

The triangle of algebraic girth 1 is shown in Figure~\ref{fig:onegirth}.  For
reflexive digraphs $H$ we show the following.

\begin{theorem}
	\label{thm:reflexive}
	Let $H$ be a reflexive digraph that contains neither a triangle of
	algebraic girth 1 nor a 4-cycle of algebraic girth 0 as subgraph. Then
	\recollong{H} admits a polynomial-time algorithm.
\end{theorem}

One intriguing property of reconfiguration problems is that they can be easy
even if the underlying decision problem is hard (and vice versa). To illustrate
this, consider the classical 3-Coloring problem, which can be stated as the
following graph homomorphism problem: Given an undirected graph $G$, is there a
homomorphism from $G$ to the complete graph on three vertices? While 3-Coloring
is \NP-complete, Cereceda et al.~have shown in~\cite{Cereceda:11} that, given two such homomorphisms,
there is a polynomial-time algorithm that decides whether there is a
step-by-step transformation between them.
Theorem~\ref{thm:c4free} implies a large class of such examples. For instance, it
is known that  there are orientations $H$ of a tree such that \csp{H} is
\NP-complete, but Theorem~\ref{thm:c4free} shows that \recollong{H} admits a
polynomial-time algorithm for any orientation $H$ of a tree. The situation is
different for reflexive graphs. Deciding if there is a homomorphism into a
fixed reflexive graph $H$ is trivial, since a homomorphism may map all
vertices of a graph to the same looped vertex of $H$. However, deciding if two
given homomorphisms to a reflexive graph $H$ admit a step-by-step
transformation turns out to be non-trivial (see Theorem~\ref{thm:reflexive}).

\begin{figure}
    \centering
    \begin{tikzpicture}[node distance={2.5cm}, thick, main/.style = {draw, circle}] 
    	\node[main] (1) {}; 
    	\node[main] (2) [below right of=1] {}; 
    	\node[main] (3) [below left of=1] {}; 
    	\node[main] (4) [below left of=2] {}; 
    	\draw [very thick,->] (1) -- (2); 
    	\draw [very thick,->] (2) -- (4); 
    	\draw [very thick,->] (1) -- (3); 
    	\draw [very thick,->] (3) -- (4); 
    \end{tikzpicture} 
    \hspace{3em}
    \begin{tikzpicture}[node distance={2.5cm}, thick, main/.style = {draw, circle}] 
    	\node[main] (1) {}; 
    	\node[main] (2) [below right of=1] {}; 
    	\node[main] (3) [below left of=1] {}; 
    	\node[main] (4) [below left of=2] {}; 
    	\draw [very thick,->] (1) -- (2); 
    	\draw [very thick,<-] (2) -- (4); 
    	\draw [very thick,->] (1) -- (3); 
    	\draw [very thick,<-] (3) -- (4); 
    \end{tikzpicture}
    \caption{The two non-isomorphic orientations of a 4-cycle with algebraic girth zero.}
    \label{fig:zerogirth}
\end{figure}
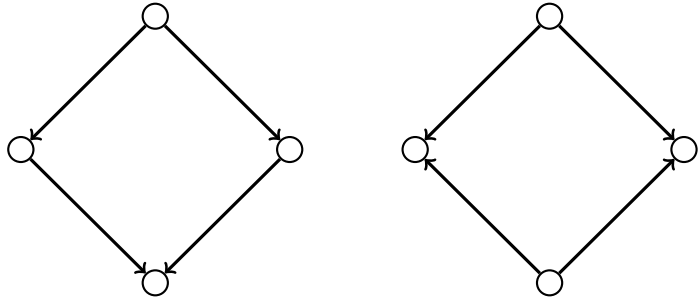

\begin{figure}
    \centering
    \begin{tikzpicture}[node distance={2.5cm}, thick, main/.style = {draw, circle}] 
    	\draw (0,0) node[main] (1) {}; 
    	\draw (3,0) node[main] (2) {}; 
    	\draw (1.5,2.598) node[main] (3) {}; 
    	\draw [very thick,->] (1) -- (2); 
    	\draw [very thick,->] (2) -- (3); 
    	\draw [very thick,->] (1) -- (3); 
    \end{tikzpicture} 
    \caption{The orientation of a 3-cycle with algebraic girth 1.}
    \label{fig:onegirth}
\end{figure}
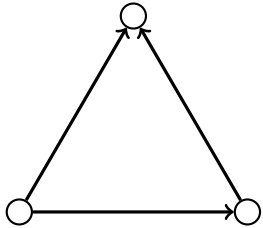

\subsection{Our results and their relation to Wrochna's algorithm}

We show that the topological approach introduced by Wrochna
in~\cite{Wrochna:20} for the reconfiguration of homomorphisms of undirected graphs
can be extended to the digraph homomorphisms. An undirected graph is called
\emph{square-free} if it does not contain a cycle on four vertices as a
subgraph. For a homomorphism $G \to H$ of directed or undirected graphs, we
refer to the image of a vertex as its \emph{color}. Two digraph homorphisms
$\alpha, \beta : G \to H$ admit a step-by-step transformation if there is a
sequence $f_1, f_2, \ldots, f_\ell$ of homomorphisms $G \to H$, such that
$\alpha = f_1$, $\beta = f_\ell$, and any two consecutive homomorphisms $f_i$,
$f_{i+1}$ differ with respect to the color of exactly one vertex. Such a
sequence is called \emph{$H$-recoloring sequence} (from $\alpha$ to $\beta$).
One key observation of Wrochna is that if $H$ is an undirected square-free
graph then, whenever the color of a vertex changes during a transformation, 
all of its neighbors must have the same color. This so-called
\emph{monochromatic neighborhood property} implies that, roughly speaking, the
color changes of a single vertex in a $H$-recoloring sequence determine those of all
other vertices. This sets the stage for a topological point of view on the
problem, which in turn results in a characterization of \emph{all} possible
sequences of color changes of a single arbitrary vertex of $G$ that generate
$H$-recoloring sequences satisfying the monochromatic neighborhood property.
This characterization allows for a polynomial-time algorithm that either finds
a suitable sequence of color changes of a vertex or concludes that no such
sequence exists.

At first sight, the crucial premise of Wrochna's algorithm seems to be the
square-freeness of $H$. But in fact the algorithm finds for any undirected
graph $H$, square-free or not, a representation of all $H$-recoloring sequences
that satisfy the monochromatic neighborhood property (possibly there is no such
sequence). In fact, Wrochna remarks the following.
\begin{remark}[{\cite{Wrochna:20}}]
    \label{rem:nostructure}
    We note that none of the proofs in this paper used any structural
    properties of H. If we consider $H$-Recoloring for any graph $H$, but only
    allow recoloring a vertex if all of its neighbors have one common color (in
    other words, a reconfiguration step is allowed only when the homotopy class
    of the mapping does not change), the same results will follow.
\end{remark}
From this remark it follows that the algorithm also works in the following
cases for undirected graphs $H$ with loops allowed: $H$ does not contain $C_4$,
$K_3$ with one loop added and $K_2$ with both loops added. 

For a loopless digraph $H$, a natural structural property that enforces the monochromatic neighborhood property of any $H$-recoloring sequence is that $H$ does not contain a 4-cycle of algebraic girth 0 (see
Figure~\ref{fig:zerogirth}). 
It is not hard to see that any $H$-recoloring sequence induces a $\bar
H$-recoloring sequence for the underlying undirected graph $\undir{H}$ of $H$.
Following the discussion of Wrochna's algorithm above, we may apply it to the
graph $\undir{H}$ and it will return a representation of all $\undir{H}$-recoloring
sequences that satisfy the monochromatic neighborhood property.
Since we know that, under the restrictions of Theorem~\ref{thm:c4free}, all
$H$-recoloring sequences satisfy the monochromatic neighborhood property, it
remains to determine whether or not there is one of the $\undir{H}$-recoloring
sequence 
that is compatible with the
orientation of the arcs of $H$. For this purpose we introduce the so-called zigzag
condition that allows us to
obtain a characterization of all possible sequences of color changes of a
particular vertex of $G$ that generate precisely $H$-recoloring sequences that
satisfy the monochromatic neighborhood property and in addition are compatible
with the orientation of the arcs of $H$ (see Theorem~\ref{thm:c4:classification}). This
characterization leads to the polynomial-time algorithm of
Theorem~\ref{thm:c4free}.

In order to prove Theorem~\ref{thm:reflexive}, we first adapt Wrochna's
algorithm to undirected reflexive graphs. The topological intuition behind the
monochromatic neighborhood property is that all $H$-recoloring sequences
satisfying this property are in the same homotopy class (see~\cite{Wrochna:20} for more details). 
For
reflexive undirected graphs we introduce another such property:
We say that an $H$-recoloring sequence satisfies the \emph{\pop property} if, 
whenever a vertex $v$ of a graph $G$ changes its color, say from $a$ to $b$, then the color of any neighbor of $v$ in $G$ is either $a$ or $b$. The natural
structural property on a reflexive undirected graph $H$ that ensures that
any $H$-recoloring sequence satisfies the \pop property is that $H$ is
triangle-free. 
We obtain a
characterization of sequences of color changes of vertices of $G$ that generate
$H$-recoloring sequences satisfying the \pop property (see
Theorem~\ref{thm:R:realwalks}). This result implies that  \recollong{H} admits a
polynomial-time algorithm for any undirected reflexive graph $H$ of girth at
least 5, which has been proved recently by Lee et al.~using other
methods~\cite{Lee:21}. For reflexive digraphs $H$, we can apply the algorithm of
Theorem~\ref{thm:R:realwalks} (almost) as a black box for the corresponding undirected graph $\undir{H}$ and then check whether any of the
$\undir{H}$ recoloring sequences is compatible with the orientation of the arcs
of $H$. Depending on which case of Theorem~\ref{thm:R:realwalks} applies,
different levels of sophistication are required, but in any case there is a
polynomial-time algorithm. We thus obtain Theorem~\ref{thm:reflexive}.

\subsection{Related work}

The complexity of \recollong{H} for undirected graphs $H$ has been studied
systematically, in particular since the work of Cereceda et
al.~\cite{Cereceda:11}, who showed that if $H$ is a $K_3$, a complete graph on
three vertices, then \recollong{H} admits a polynomial-time algorithm, despite
\csp{K_3} (``3-Coloring'') being \NP-complete. Wrochna generalized this result
in~\cite{Wrochna:20} showing that \recollong{H} admits a polynomial-time
algorithm if $H$ is loopless and square-free.  Brewster and al.~gave a
complexity classification of \recollong{H} for the class of  circular cliques
$C_{p, q}$~\cite{Brewster:16}. Note that their polynomial-time algorithm for $2
\leq p/q < 4$ includes graphs $H$ that are not square-free. Recently, Lee et
al.~adapted Wrochna's algorithm to the case that $H$ is reflexive and has girth
at least 5~\cite{Lee:21}.  On the negative side, it is known that \recollong{H}
is \PSPACE-complete if $H$ is a clique on at least four vertices~\cite{BC:09},
a circular clique $C_{p, q}$ where $p/q \geq 4$~\cite{Brewster:16}, a wheel on
a $k$-cycle, where $k \geq 3$ and $k \neq 4$~\cite{Lee:20}, or a
quadrangulation with certain properties~\cite{Lee:20}.
    
We are aware of two results for \recollong{H} for digraphs $H$. The first one is
by Brewster et al., who showed that \recollong{H} admits a polynomial-time
algorithm if $H$ is a reflexive digraph cycle that does not contain a 4-cycle
of algebraic girth 0~\cite{Brewster:17}. In spirit this algorithm uses the
topological approach of Wrochna that reduces the task of finding $H$-recoloring
sequences to finding vertex walks in $H$. 
Secondly, Dochterman and Singh study the \Hom-complex for digraphs $G$ and $H$
and show that it is connected (in the topological sense) if $H$ is a transitive
tournament $T_n$~\cite{Dochtermann:21} on $n$ vertices. From this they conclude
that any instance of \recollong{T_n} is a \yes-instance and give a
polynomial-time algorithm that finds a $T_n$-recoloring sequence. The algorithm
is simple and does not need any topological tools. It boils down to the fact
that a homomorphism of an acyclic digraph into a tournament corresponds to a
linear extension of a partial order. To reconfigure one linear extension into
another, we may take the last element where the two linear extensions disagree and assign to the vertex with the smaller image the larger image (with respect to
the total order).

Further results are known for \recollong{\mathcal{H}} for a relational
structure $\mathcal{H}$ on a Boolean domain. This problem corresponds to the
reconfiguration of satisfying assignments of Boolean formulas.
In~\cite{Gopalan:06}, Gopalan et al.~provide a complexity trichotomy for this
problem, characterizing the relations for which Boolean satisfiability
reconfiguration is in \P, \NP-complete, or \PSPACE-complete. 
Another popular theme with a different flavor is the reconfiguration of subgraphs, which may be
considered to be homomorphisms (injective or not) from a fixed graph $H$ to a
given graph $G$. In this context, Ito et al.~showed that reconfiguring directed paths
is \PSPACE-complete. Often however the graph $H$ is not fixed, but any
subgraph of a certain ``shape'', e.g., any tree, would be acceptable.  We refer
the readers to the survey of Nishimura~\cite{Nishimura:18} for an overview of
known results in this direction, in particular those on the reconfiguration of
independent sets, as well as~\cite{Tesshu:20} for results on spanning and
induced subgraphs. A general introduction to reconfiguration problems that
also discusses their relation to combinatorial games and puzzles can be found
in~\cite{Heuvel:13}.

\subsection{Organization}

Section~\ref{sec:preliminaries} contains notation and basic definitions that are
needed to prove Theorems~\ref{thm:c4free} and~\ref{thm:reflexive}. In
Section~\ref{sec:homgraph} we will discuss the relation of \recollong{H} to the
Hom-graph. Section~\ref{sec:monochromatic} contains the polynomial-time
algorithm for \recollong{H}, where $H$ is a loopless digraph that contains no
4-cycle of algebraic girth 0 (Theorem~\ref{thm:c4free}).
Section~\ref{sec:push-pull} is divided into two parts. In
Section~\ref{sec:unreflexive} we adapt Wrochna's algorithm to undirected
reflexive graphs. We use this algorithm in Section~\ref{sec:R:directed} in
order obtain a polynomial-time algorithm for \recollong{H} for reflexive
digraphs $H$ that contain neither a 4-cycle of algebraic girth 0 nor a triangle
of algebraic girth 1 (Theorem~\ref{thm:reflexive}).

\section{Preliminaries}
\label{sec:preliminaries}

A directed graph (digraph) is a pair $(V(G), A(G))$ where $V(G)$ is a finite
set of vertices and $A(G) \subseteq V(G) \times V(G)$ are \textit{arcs}. We
write $u \rightarrow v$ when $uv \in A(G)$. We say that a digraph $G$ is
\emph{symmetric} if $vu \in A(G)$ whenever $uv \in A(G)$. A digraph $G$ is
\emph{reflexive} if $uu \in A(G)$ for each vertex $u \in V(G)$. We interpret a
symmetric digraph as \emph{undirected graph} and think of two edges $\{uv,
vu\}$ as undirected edge, which we also write as $uv$ since it should be clear
from the context whether we refer to a directed or undirected edge. We write
$E(G)$ for the set of undirected edges of a symmetric graph $G$. For any
digraph $G$, we canonically associate to $G$ an unoriented graph $\bar{G}$
where $V(\bar{G}) = V(G)$ and $uv \in E(\bar{G})$ if $u \rightarrow v$ or $v
\rightarrow u$. Let $G$ be a digraph. The \emph{in-neighborhood} (resp.,
\emph{out-neighborhood}) of a vertex $v \in V(G)$ is given by $N^-_G(v) := \lbrace w
\in V(G) \mid w \rightarrow v \rbrace$ (resp., $ N^+_G(v) := \lbrace w \in V(G)
\mid v \rightarrow w \rbrace$).  If $G$ is symmetric (undirected), the
neighborhood $N_G(v)$ of a vertex $v \in V(G)$ is the set of vertices adjacent
to $v$ in $G$, that is, $N_G(v) := \lbrace w \in V(G) \mid vw \in E(G)
\rbrace$. 
Let $G$ and $H$ be digraphs. A homomorphism $\phi : G \to H$ or ($H$-coloring
of $G$) is a map $V(G) \to V(H)$ that preserves arcs, that is, for each $u
\rightarrow v$, we have $\phi(u) \rightarrow \phi(v) $. Similarly, for
undirected graphs $G$ and $H$, a homomorphism $\phi \colon G \to H$ is a map
$V(G) \to V(H)$ that preserves edges (but not necessarily non-edges).  A
homomorphism $\alpha \colon G \to H$ also induces a homomorphism
$\undir{\alpha} \colon \undir{G} \to \undir{H}$.

A \emph{walk} $W$ in an undirected graph $G$ is a sequence of consecutive edges $W
=(v_1v_2)(v_2 v_3)\ldots (v_{n-1} v_n)$. The reverse walk $W^{-1}$ of a walk $W
=(v_1v_2)(v_2v_3)\ldots (v_{n-1}v_n)$ is the
walk $W^{-1} = (v_n v_{n-1}) \ldots (v_2 \, v_1)$. The length $|W|$ of $W$
the number of edges of $W$. A \emph{cycle} $C$ is a closed walk, i.e., a walk
such that $v_1 = v_n$.  A walk in digraph $G$ is a walk in $\undir{G}$. The
\textit{algebraic girth} of a cycle $C$ in a digraph $G$ is the absolute value
of the number of forward arcs minus the number of backward arcs. We say that a
graph or a digraph $G$ is \emph{connected} if for any two vertices $u,v \in
V(G)$ there is a walk from $u$ to $v$ in $G$. A walk $W =
(v_1v_2)\ldots(v_{n-1}v_n)$ in a digraph is \emph{directed} if $v_i \rightarrow
v_{i+1}$ for all $1 \leq i \leq n-1$. The walk $W$ is \emph{symmetric} if both
$W$ and $W^{-1}$ are directed.  We denote the empty
walk by $\varepsilon$.

\paragraph{The fundamental groupoid.} Let $H$ be an undirected or directed graph. Given a walk $W = (v_1 \, v_2)(v_2 \, v_3) \ldots (v_{n-1} \, v_n)$ in $H$, we call \emph{reduction} the two following operations (see Figure~\ref{fig:reduction}):
\begin{itemize}
\item The operation of deleting $(v_i \, v_{i+1})(v_{i+1} \, v_{i+2})$ from $W$ if $v_i = v_{i+2}$ and $1 \leq i \leq n-2$
\item The operation of deleting $(v_i \, v_{i+1})$ from $W$ if $v_i = v_{i+1}$ and $1 \leq i \leq n-1$. Note that this operation requires a loop on $v_i$, so it applies in particular if $H$ is reflexive.
\end{itemize}

We say that $W$ is \textit{reduced} if none of the two operations above can be
applies. That is, for $1 \leq i \leq n-2$, we have $v_{i+2} \neq v_i$ and for
$1 \leq i \leq n-1$, we have $v_{i+1} \neq v_i$. We can \emph{reduce} a walk
$W$ by iteratively applying reductions on it; we can easily see that by doing
so we obtain a unique reduced walk. By considering two walks to be equivalent
if they reduce to the same walk, we obtain an equivalence relation~$\sim$ on
the walks in $H$. The fundamental groupoid $\pi(H)$ is the set of all
equivalence classes of walks in $H$ under~$\sim$. Its groupoid operation is the
concatenation $\cdot$ of walks and its neutral element is the empty walk
$\varepsilon$. For any walk $W= (v_1 \, v_2) \ldots (v_{n-1} \, v_n)$, the
inverse of (the class of) $W$ in $\pi(H)$ is (the class of) the reversed walk
$W^{-1} = (v_n \, v_{n-1}) \ldots (v_2 \, v_1)$ since both $WW^{-1}$ and
$W^{-1}W$ reduce to $\epsilon$. In the next sections of this paper, we will
write $W_1 = W_2$ \emph{in} $\pi(H)$ if $W_1 \sim W_2$, that is, $W_1$ and
$W_1$ reduce to the same walk.

\begin{figure}
\begin{tikzpicture}
	\begin{scope}
	\dbubblegraph;
	\node[G1] (v0) at ($(ha0) + (60:4pt)$) {};
	\node[G] (v1) at ($(ha1) + (60:5pt)$) {};
	\node[G] (v2) at ($(ha2) + (90:5pt)$) {};
	\node[G] (v3) at ($(hb2) + (-30:5pt)$) {};
	\node[G] (v4) at ($(hb1) + (-150:5pt)$) {};
	\node[G] (v5) at ($(hb0) + (150:5pt)$) {};
	\node[G] (v6) at ($(hb0) + (30:5pt)$) {};
	\node[G] (v7) at ($(hb1) + (30:5pt)$) {};
	\node[G] (v8) at ($(hb2) - (-30:5pt)$) {};
	\node[G] (v9) at ($(ha2) - (90	:5pt)$) {};
	\node[G] (v10) at ($(ha1) - (60:5pt)$) {};
	\node[G] (v11) at ($(ha0) - (60:5pt)$) {};
	\node[G] (v12) at ($(hb0) + (-120:5pt)$) {};
	\node[G1] (v13) at ($(hb5)$) {};
	\draw[Gve] (v0)--(v1)--(v2)--(v3)--(v4)--(v5)--(v6)--(v7)--(v8)--(v9)--(v10)--(v11)--(v12)--(v13);
	\node at (2,1.5) {\Large $W_1$};
	\end{scope}
	\begin{scope}[shift={(5.5,0)}]

	\dbubblegraph;
	\node[G1] (v0) at ($(ha0)$) {};
	\node[G] (v1) at ($(hb0)$) {};
	\node[G1] (v2) at ($(hb5)$) {};
	\draw[Gve] (v0)--(v1)--(v2);
	\node at (2,-1.5) {\Large $W_2$};
	\end{scope}
	
	\begin{scope}[shift={(11,0)}]
	\dbubblegraph;
	\node[G1] (v0) at ($(ha0)$) {};
	\node[G] (v1) at ($(ha1)$) {};
	\node[G] (v2) at ($(ha2)$) {};
	\node[G] (v3) at ($(hb2)$) {};
	\node[G] (v4) at ($(hb1)$) {};
	\node[G] (v5) at ($(hb0)$) {};
	\node[G1] (v6) at ($(hb5)$) {};
	\draw[Gve] (v0)--(v1)--(v2)--(v3)--(v4)--(v5)--(v6);
	\node at (2.2,1.2) {\Large $W_3$};
	\end{scope}
\end{tikzpicture}

\caption{Consider the grey graph $H$ and the three walks $W_1, W_2$ and $W_3$ in $H$ (their successive vertices are drawn in black). The walk $W_1$ reduces to $W_2$. The walk $W_3$ is already reduced.}
\label{fig:reduction}
\end{figure}
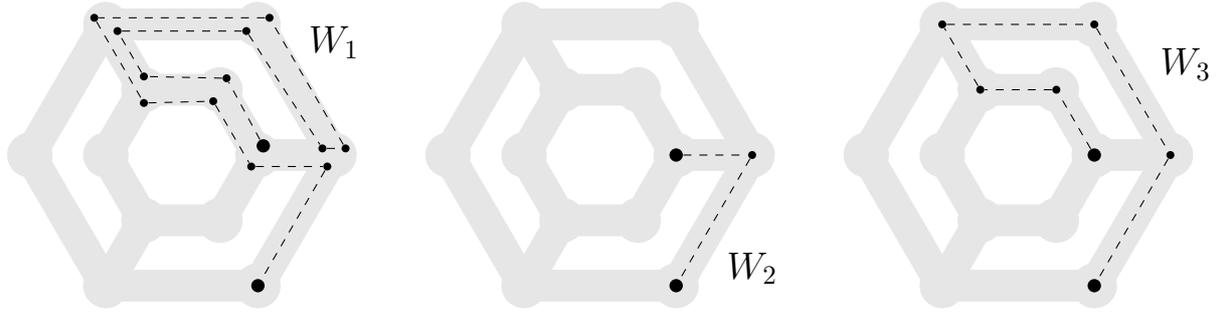

\paragraph{Cyclic reduction.} We say that closed walk $C = (v_1v_2)\ldots(v_{n-1}v_1)$ is \emph{cyclically reduced} if it is reduced and additionally $v_2 \neq v_{n-1}$. We can \emph{cyclically reduce} any reduced closed walk $C$ by iteratively deleting from it both its first and last edges while one is the inverse of the other. This operation leads to a unique decomposition $C = A^{-1} C_0 A$ where $A$ is the sequence of deleted edges of $C$ and $C_0$ is cyclically reduced.

\paragraph{The $H$-recoloring problem.} Let $G$ and $H$ be digraphs. Recall that two digraph homomorphisms $\alpha, \beta : G \to H$ admit an \emph{$H$-recoloring sequence} if for some $\ell$ there are digraph homomorphisms $f_1, f_2, \ldots, f_\ell : G \to H$, such that $\alpha = f_1$, $\beta = f_\ell$ and for $1 \leq i < \ell$, the two homorphisms $f_i$ and $f_{i+1}$ differ with respect to the color of a single vertex of $G$. The problem \recol{H} asks whether, given a graph $G$ and two digraph homomorphisms $\alpha, \beta : G \to H$, the homomorphisms $\alpha$ and $\beta$  admit an $H$-recoloring sequence.
Equivalently, the problem \recol{H} can be defined as a reachability problem on the following (undirected) \emph{reconfiguration graph} $\mathcal{R}_H(G)$. The vertex set of $\mathcal{R}_H(G)$ is the set of homomorphisms from $G$ to $H$ and two vertices of $\mathcal{R}_H(G)$ are adjacent if the corresponding homomorphisms differ with respect to the image of exactly one vertex of $G$. Hence, \recol{H} asks whether there is a path from $\alpha$ to $\beta$ in  $\mathcal{R}_H(G)$.

\paragraph{In this entire paper we assume that $G$ and $H$ are directed or undirected, (weakly) connected graphs, with at least two vertices.} One can see that assuming the connectivity of $G$ and $H$ imposes no restrictions, since if $G$ is not connected, we may consider the reconfiguration of each connected component of $G$ separately. If $H$ is not connected, observe that any connected component of $G$ maps to a connected component of $H$.

\section{\recollong{H} and \Hom-graphs}
\label{sec:homgraph}

In this section we will discuss different natural adjacency relations that have
been considered for digraph homomorphisms.  Our main technical goal is to show
that if $H$ does not contain a cycle of algebraic girth 0 then \recollong{H} is
equivalent to reachability in the graph $\Hom_1(G, H)$, which will be defined
shortly. In the setting of undirected reflexive graphs this has been observed
in~\cite{Lee:21} and the proof essentially carries over to digraphs. Similar to~\cite{Lee:21},
we use this result in order to generalize an algorithm for reflexive
instances of \recollong{H}, where $H$ is reflexive and triangle-free, to all
instances to obtain Theorem~\ref{thm:reflexive}. 

Let us consider an adjacency relation on graph homomorphisms that is different from the one used to define $\mathcal{R}_H(G)$. It yields the so-called \Hom-graph $\Hom(G, H)$, whose vertices are
homomorphism from $G$ to $H$ and two homomorphisms $\phi, \psi : G \to H$ are
adjacent if for each edge $uv$ of $G$, we have that $\phi(u)
\psi(v)$ is an edge of $H$. 
It is known that for simple, undirected graphs (symmetric digraphs),
the graphs $\mathcal{R}_H(G)$ and $\Hom(G,H)$ have the same connected
components (although their edge sets may be different). For reflexive
undirected graphs or digraphs (with or without loops) this is not the case in
general (an example is given below).  That being said, Brewster and
al.~generalized their results for \recollong{H} where $H$ is a reflexive digraph cycle from~\cite{Brewster:17} to the corresponding recoloring problems
in the \Hom-graph in~\cite{Brewster:21} (with or without restriction on the number of vertex that can change color at the same step). 

Homomorphisms that are adjacent in $\Hom(G, H)$ may differ with respect to the
colors of more than one vertex of $G$. In order to get something more similar
to the graph $\mathcal{R}_H(G)$, we may consider the subgraph $\Hom_1(G, H)$ of
$\Hom(G, H)$, where two homomorphisms are adjacent if they are in $\Hom(G, H)$
and they differ with respect to the color of exactly one vertex of $G$. Let
\recolun{H} be the problem that asks whether, given a digraph $G$ (symmetric or
not) and two homomorphisms $\phi, \psi : G \to H$, there is a path connecting
$\phi$ and $\psi$ in $\Hom_1(G, H)$. To illustrate the differences between the
graphs $\mathcal{R}_H(G)$, $\Hom_1(G, H)$ and $\Hom(G, H)$, we consider the
example shown in Figure~\ref{fig:homvsrecol}. The graph $G$ is the oriented
cycle shown in black. The graph \loops{H} is an orientation of the complete graph on
five vertices with a loop on each vertex and the graph $G$ is an oriented 5-cycle. Two homomorphisms
$\phi, \psi : G \to H$ are shown figures~\ref{fig:homvsrecol:phi}
and~\ref{fig:homvsrecol:psi}, respectively. The two homomorphisms are adjacent
in $\Hom(G, \loops{H})$, but not in $\Hom_1(G, \loops{H})$ or in $\mathcal{R}_\loops{H}(G)$, since only a
single vertex can change its color. Let $G'$ be the graph with a single looped vertex $u$. We
observe that $u$ can change its color freely to any other vertex of \loops{H} in
$\mathcal{R}_{\loops{H}}(G')$, but in $\Hom(G', \loops{H})$ and $\Hom_1(G',
\loops{H})$ it can only move to a neighbor in \loops{H}.  

We denote by \recolun{H} the problem of deciding whether for two given
homomorphisms $\phi, \psi$ from a graph $G$ to a fixed template graph $H$ there
is a path in $\Hom_1(G, H)$.  We show that under the conditions imposed by
Theorems~\ref{thm:c4free} the problems \recol{H} and \recolun{H} are
equivalent. The proof is essentially the same as the one given in~\cite{Lee:21} for
reflexive undirected graphs.

\begin{lemma} 
    \label{lem:equirecol} 
    Let $H$ be a reflexive digraph. If $H$ does not contain a $4$-cycle of
    algebraic girth $0$ then \recol{H} and \recolun{H} are polynomially
    equivalent.
\end{lemma}
\begin{proof}
	Let $\alpha,\beta \colon G \to H$. By the definition of $\Hom_1(G, H)$ it
	is clear that any path from $\phi$ to $\psi$ in $\Hom_1(G, H)$ is also
	a path in $\mathcal{R}_H(G)$, so in particular a \yes-instance of
	\recolun{H} is also a \yes-instance of \recollong{H}.

	Conversely, consider any step $\phi, \psi$ of an $H$-recoloring sequence
	from $\alpha$ to $\beta$ (with respect to $\mathcal{R}_H(G)$. We show how to obtain in polynomial time a
	path from $\phi$ to $\psi$ in $\Hom_1(G, H)$. Assume that $\phi$ and
	$\psi$ differ with respect to the color of a single vertex of $G$, say
	$u$, and let $a = \phi(u)$ and $b = \psi(u)$. 
	If $a$ and $b$ are
	adjacent in $H$ then $\phi$ and $\psi$ are adjacent in $\Hom_1(G, H)$
	and we are done.
	Recall that $u$ is not isolated since $G$ is weakly connected and has
	at least two vertices. Thus, if $u$ has no loop then $\phi$ and $\psi$ are
	adjacent in $\Hom_1(H)$ and we are done.  
	Otherwise, consider a neighbor $v$ of $u$. 
	Since $\phi$ and $\psi$ agree on every vertex of $G$ except $u$, we have that $h =
	\phi(v) = \psi(v)$ and $h$ is a common neighbor of $a$ and $b$ in $H$,
	thus distinct from $a$ and $b$ as they are non-adjacent. If any
	neighbor $w$ of $u$ has a color different from $h$ with respect to
	$\phi$ and $\psi$ then again this color is a common neighbor of $a$ and
	$b$, distinct from $a$ and $b$. Therefore, $H$ contains a $4$-cycle of
	algebraic girth $0$ (see Figure~\ref{fig:zerogirth}), which contradicts
	our assumptions on $H$. Hence, all neighbors of $u$ have color $h$. Let 
	\[
		\phi'(w) = 
		\begin{cases}
			h & \text{if $w = u$, and}\\
			\phi(w) &   \text{otherwise ($w \in V(G) - u$).}
		\end{cases}
	\]
	Since $\phi'$ maps all arcs that are incident to $u$ to the same loop
	on $h$ and $\phi$ is a homomorphisms, we have that $\phi'$ is a
	homomorphism from $G$ to $H$.  Furthermore, $\phi, \phi', \psi$ is a
	path from $\phi$ to $\psi$ in $\Hom_1{G, H}$. Hence, we apply this
	reasoning to all steps of the initial $H$-recoloring sequence from
	$\alpha$ to $\beta$ to obtain a path from $\phi$ to $\psi$ in
	$\Hom_1(G, H)$ that is at most twice as long as the initial one.
\end{proof}

\begin{figure}
	\subcaptionbox{Illustration of the homomorphisms $\phi : G \to H$.\label{fig:homvsrecol:phi}}[0.5\textwidth] {
                \centering
		\begin{tikzpicture}
		    \begin{scope}
			\dstargraph;
			\node[G,label=0:$v_0$] (v0) at ($(ha0)$){};
			\node[G,label=72:$v_1$] (v1) at ($(ha1)$){};
			\node[G,label=144:$v_2$] (v2) at ($(ha2)$){};
			\node[G,label=-144:$v_3$] (v3) at ($(ha3)$){};
			\node[G,label=-72:$v_4$] (v4) at ($(ha4)$){};
			\draw[Ge] (v0)--(v1) node[sloped,pos=0.55,allow upside down]{\arrowInG}; ;
			\draw[Ge] (v1)--(v2) node[sloped,pos=0.55,allow upside down]{\arrowInG}; ;
			\draw[Ge] (v2)--(v3) node[sloped,pos=0.55,allow upside down]{\arrowInG}; ;
			\draw[Ge] (v3)--(v4) node[sloped,pos=0.55,allow upside down]{\arrowInG}; ;
			\draw[Ge] (v4)--(v0) node[sloped,pos=0.55,allow upside down]{\arrowInG}; ;
			\node at (1.7,-1.4) {\Large$\phi$};
		    \end{scope}
	    \end{tikzpicture}
    }
    \subcaptionbox{Illustration of the homomorphisms $\psi : G \to H$.\label{fig:homvsrecol:psi}}[0.5\textwidth] {
                \centering
		\begin{tikzpicture}
			\begin{scope}[shift={(5.5,0)}]
				\dstargraph;
				\node[G,label=0:$v_4$] (v0) at ($(ha0)$){};
				\node[G,label=72:$v_0$] (v1) at ($(ha1)$){};
				\node[G,label=144:$v_1$] (v2) at ($(ha2)$){};
				\node[G,label=-144:$v_2$] (v3) at ($(ha3)$){};
				\node[G,label=-72:$v_3$] (v4) at ($(ha4)$){};
				\draw[Ge] (v0)--(v1) node[sloped,pos=0.55,allow upside down]{\arrowInG}; ;
				\draw[Ge] (v1)--(v2) node[sloped,pos=0.55,allow upside down]{\arrowInG}; ;
				\draw[Ge] (v2)--(v3) node[sloped,pos=0.55,allow upside down]{\arrowInG}; ;
				\draw[Ge] (v3)--(v4) node[sloped,pos=0.55,allow upside down]{\arrowInG}; ;
				\draw[Ge] (v4)--(v0) node[sloped,pos=0.55,allow upside down]{\arrowInG}; ;
				\node at (1.7,-1.4) {\Large$\psi$};
			\end{scope}
		\end{tikzpicture}
	}
	\caption{A reflexive digraph \loops{H}, which is an orientation of a complete graph (shown in grey) and two homomorphisms $\phi, \psi$ mapping a directed 5-cycle $G$ to $H$ (shown in black in~\ref{fig:homvsrecol:phi} and~\ref{fig:homvsrecol:psi}, respectively). The homomorphisms $\phi$ and $\psi$ are adjacent in $\Hom(G, H)$ but not in  $\Hom_1(G,H)$ or $\mathcal{R}_H(G)$.}
	\label{fig:homvsrecol}
\end{figure}
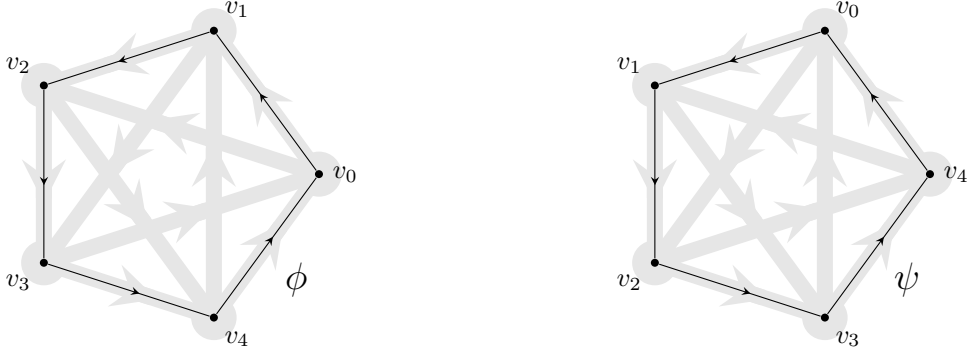

\section{General tools}
\label{sec:general}

This section introduces tools that will be used in the proof of both Theorem~\ref{thm:c4free} and Theorem~\ref{thm:reflexive}. Let $G$ and $H$ be both undirected graphs or digraphs, with loops or not. Let $\alpha, \beta : G \to H$ be graph homomorphisms and let $q \in V(G)$. The topological validity condition introduced by Wrochna in~\cite{Wrochna:20} can be stated in our very general setting as follows;
please refer to~\cite{Wrochna:20} for an explanation of the topological intuition behind it. 

\begin{definition}
A walk  $Q \in \pi(H)$ from $\alpha(q)$ to $\beta(q)$ is
\emph{topologically valid} with respect to $\alpha$, $\beta$, and $q$ if for any closed walk $C$ from $q$ to $q$ in $G$, we have $\beta(C) = Q^{-1} \cdot \alpha(C) \cdot Q$ in $\pi(H)$.
\end{definition}

Figure~\ref{fig:topval} gives an example of walks that are topologically valid and of walks that are not in the simple case when $G$ is a cycle $C$ (say a closed walk from $q$ to $q$). The topological validity of any walk $Q$ can therefore be easily checked by computing $Q^{-1} \alpha(C) Q$ and $\beta(C)$. $Q$ is topologically valid if and only if those reduce to the same walk in $\pi(H)$. So in the example, $Q_1$ and $Q_2$ are topologically valid but $Q_3$ is not.

\begin{figure}
\begin{tikzpicture}
    \begin{scope}
	    \dtwosqgraph;
	    \node[G1,label=-60:$q$] (v0) at ($(hb3)-(90:4pt)$) {};	
	    \node[G,label=180:$v_1$] (v1) at ($(hb2)-(0:4pt)$) {};
	    \node[G, label=90:$v_2$] (v2) at ($(hb3)+ (90:4pt)$) {};
    	\node[G,label=180:$v_3$] (v3) at ($(hb0)$) {};
    	\node[G] (v4) at ($(hb1)$) {};
    	\node[G] (v5) at ($(hb2)+ (0:4pt)$) {};
    	\draw[Ge] (v0)--(v1)--(v2)--(v3)--(v4)--(v5)--(v0);
    	\node at (1.5,1.5) {\Large$\alpha$};
    	\node at (4.2,1) {$\alpha(C)$};
    \end{scope}
	
	\begin{scope}[shift={(6.5,0)}]
		\dtwosqgraph;
		\node[G1,label=120:$q$] (v0) at ($(hb0)-(0:4pt)$) {};
		\node[G,label=90:$v_1$] (v1) at ($(hb1)+(90:4pt)$) {};
		\node[G, label=0:$v_2$] (v2) at ($(hb0)+(0:4pt)$) {};
		\node[G,label=-90:$v_3$] (v3) at ($(hb1)-(90:4pt)$) {};
		\node[G] (v4) at ($(hb2)$) {};
		\node[G] (v5) at ($(hb3)$) {};
		\draw[Ge] (v0)--(v1)--(v2)--(v3)--(v4)--(v5)--(v0);
		\node at (1.5,1.5) {\Large$\beta$};
		
	\end{scope}
\end{tikzpicture}

\begin{tikzpicture}[scale=0.8]
	\dtwosqgraph;
	\node[G1] (q0) at ($(hb3)$) {};
	\node[G1] (q1) at ($(hb0)$) {};
	\draw[Gve] (q0)--(q1) node[sloped,pos=0.65,allow upside down]{\arrowInG}; ; 
	\node at (1.5,-1.5) {$Q_1$};
	
	\begin{scope}[shift={(6,0)}]
		\dtwosqgraph;
		\node[G1] (q0) at ($(hb3)$) {};
		\node[G1] (q1) at ($(hb2)$) {};
		\node[G1] (q2) at ($(hb1)$) {};
		\node[G1] (q3) at ($(hb0)$) {};
		\draw[Gve] (q0)--(q1) node[sloped,pos=0.65,allow upside down]{\arrowInG}; ; 
		\draw[Gve] (q1)--(q2)--(q3);
		\node at (1.5,-1.5) {$Q_2$};
		
	\end{scope}
	\begin{scope}[shift={(12,0)}]
		\dtwosqgraph;
		\node[G1] (q0) at ($(hb3)$) {};
		\node[G1] (q1) at ($(hb0)-(90:4pt)$) {};
		\node[G1] (q2) at ($(ha0)-(90:4pt)$) {};
		\node[G1] (q3) at ($(ha3)$) {};
		\node[G1] (q4) at ($(ha2)$) {};
		\node[G1] (q5) at ($(ha1)$) {};
		\node[G1] (q6) at ($(ha0)+(90:4pt)$) {};
		\node[G1] (q7) at ($(hb0)+(90:4pt)$) {};
		\draw[Gve] (q0)--(q1) node[sloped,pos=0.65,allow upside down]{\arrowInG}; ; 
		\draw[Gve] (q1)--(q2)--(q3)--(q4)--(q5)--(q6)--(q7);
		\node at (1.5,-1.5) {$Q_3$};
		
	\end{scope}
\end{tikzpicture}
\caption{$H$ is the gray graph on $8$ vertices. $G$ is the $6$-cycle in black. $Q_1,Q_2$ and $Q_3$ are three walks from $\alpha(q)$ to $\beta(q)$. $Q_1$ and $Q_2$ are topologically valid but $Q_3$ is not (Intuitively, $Q_1$ and $Q_2$ pull the cycle of $G$ so that it turns around the right square of $H$, $Q_3$ would pull $G$ so that it becomes stretched around the left square of $H$).}

\label{fig:topval}

\end{figure}
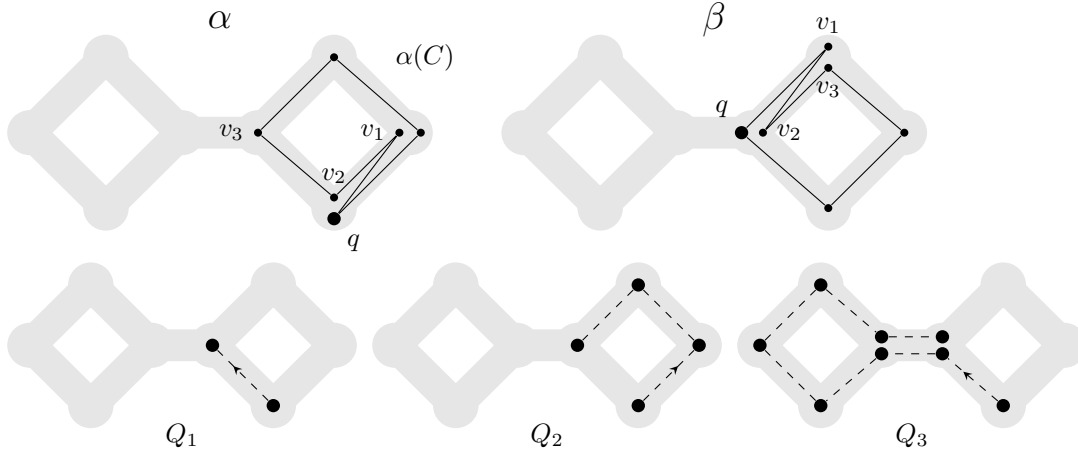

In~\cite{Wrochna:20}, Wrochna gives a classification of all topologically valid walks in the setting of square-free undirected graphs. However, we claim that this classification also works in our more general setting of digraphs, since:
\begin{itemize}
\item Orientation (when $G$ and $H$ are digraphs) plays no role for topological validity.
\item It is very clear in~\cite{Wrochna:20} that the proofs of this classification theorem do not use any structural hypothesis of $H$.
\end{itemize}

Hence, the proof of the next theorem follows, nearly word by word, from the proof of Theorem 7.5 in \cite{Wrochna:20} (also see \ref{rem:nostructure} from \cite{Wrochna:20})

\begin{theorem}[Wrochna,~\cite{Wrochna:20}] 
    \label{thm:T:Wrochna} 
    Let $\Pi$ be the set of all topologically valid walks for $\alpha,\beta,q$. One of the following holds:
	\begin{enumerate}
		\item $\Pi = \emptyset$. \label{itm:T:nowalk}
		\item $\Pi = \lbrace Q \rbrace$ for some $Q \in \pi(H)$. \label{itm:T:onlyQ}
		\item $\Pi = \lbrace R^n P \mid n \in \mathbb{Z} \rbrace$ for some $R,P \in \pi(H)$. \label{itm:T:walksRnP}
		\item $\Pi$ contains all reduced walks from $\alpha(q)$ to $\beta(q)$. \label{itm:T:allwalks}
	\end{enumerate}
	
Furthermore, we can determine in time $O(|V(G)|\cdot|E(G)| + |E(H)|)$ which case holds and output $Q$ or $R,P$ in cases~\ref{itm:T:onlyQ} and~\ref{itm:T:walksRnP} such that $|Q|,|R|,|P|$ are bounded by the total running time $O(|V(G)|\cdot|E(G)| + |E(H)|)$.
Case~\ref{itm:T:allwalks} happens if and only if $\alpha(C) = \beta(C) = \varepsilon$ for all closed walks $C$ in $G$.
\end{theorem}

Let $q \in V(G)$ and $Q$ a walk from $\alpha(q)$ to $\beta(q)$. For all $v \in V(G)$, choose a walk $W_v$ from $q$ to $v$ (such a walk exists because $G$ is assumed to be connected). We say that $Q$ and the system of walks $(W_v)_{v \in V(G)}$ \emph{generates} the walks $S_v:= \alpha(W_v)^{-1} \cdot Q \cdot \beta(W_v)$. If $Q$ is topologically valid for $\alpha$, $\beta$, and $q$, then the following lemma ensures that none of the vertex walks $S_v$ depends on the choice of $W_v$.

\begin{lemma}[Lemma 4.4 in \cite{Wrochna:20}] 
	\label{lem:equigen} 
	If $Q$ is topologically valid for $\alpha$, $\beta$, and $q$ then, for
	any vertex $v \in V(G)$ and any two walks $W_1,W_2$ from $q$ to $v$ in
	$G$, we have $\alpha(W_1)^{-1} \cdot Q \cdot \beta(W_1) = \alpha(W_2)^{-1} \cdot Q\cdot
	\beta(W_2)$ in $\pi(H)$.
	
	In other words, the walks $S_v$ generated by $Q$ and a system of walks $(W_v)_{v \in V(G)}$ do not depend of the choice of the walks $W_v$ when $Q$ is topologically valid.
\end{lemma}

The proof in~\cite{Wrochna:20} perfectly applies to our context:

\begin{proof} 
	Note that $W_1W_2^{-1}$ is a closed walk in $G$ starting and ending at $q$, so by the topological validity of $Q$,
\[ 
\beta(W_1W_2^{-1}) = Q^{-1} \cdot \alpha(W_1W_2^{-1}) \cdot Q \; \text{in } \pi(H) \enspace.  
\] 
Therefore,
\[ 
	\beta(W_1)\beta(W_2)^{-1} = Q^{-1} \alpha(W_1)\alpha(W_2)^{-1} \cdot Q\enspace.
\]
Multiplying from the left by $\alpha(W_1)^{-1} \cdot Q$ and from the right by $\beta(W_2)$ gives the desired result.
\end{proof}

\begin{lemma} \label{lem:symwalks} Let $V' \subset V(G)$ be any nonempty set of vertices of $G$.  Let $q \in V'$ and for each $v \in V'$, let $W_v$ be a walk from $q$ to $v$ of length $|W_v| \leq |V(G)|$. Then the set of walks from $\alpha(q)$ to $\beta(q)$ that generate symmetric vertex walks with the system $(W_v)_{v \in V'}$ on $V'$ is one of the following:
\begin{enumerate}
    \item $\emptyset$. \label{itm:sym:empty}
    \item $\lbrace Q \rbrace$ for some symmetric walk $Q$ of length at most $2|V(G)|$. \label{itm:sym:onlyQ}
    \item All symmetric walks from $\alpha(q)$ to $\beta(q)$. \label{itm:sym:all}
\end{enumerate}
Furthermore, we can decide in time $O(|V(G)|\cdot|E(G)|)$ which case holds and output $Q$ in Case~\ref{itm:sym:onlyQ}.
\end{lemma}
\begin{proof} Assume that $Q$ and $(W_v)_{v \in V(G)}$ generate symmetric vertex walks. Then, in particular, $Q$ is symmetric. Note that $\alpha(W_v)$ and $\beta(W_v)$ may have non symmetric edges, but since $Q$ generates only symmetric walks and $S_v$ is the reduced walk obtained from $\alpha(W_v)^{-1} \cdot Q \cdot \beta(W_v)$, we have that the non-symmetric edges of $\alpha(W_v)$ and of $\beta(W_v)$ must be the same and appear in the same order. Suppose that $e_1, e_2, \ldots, e_p$ are the non-symmetric edges of $\alpha(W_v)$ and $\beta(W_v)$. Then we can write $\alpha(W_v) = A_1 e_1 A_2 \ldots A_pe_pA_{p+1}$ and $\beta(W_v) = B_1 e_1 B_2 \ldots B_p e_p B_{p+1}$, where $A_1, \ldots, A_{p+1}$ and $B_1, \ldots, B_{p+1}$ are symmetric. Note that if $\alpha(W_v)$ and $\beta(W_v)$ are symmetric then $\alpha(W_v) = A_1$ and $\beta(W_v) = B_1$. Since all non-symmetric edges cancel in $S_v$, we obtain
\[ 
    A_p^{-1} \cdots A_2^{-1} e_1 A_1^{-1} Q B_1 e_1 B_2 \cdots B_p = \varepsilon \enspace. 
\]
So in particular, $Q = A_1 B_1^{-1}$ (so $|Q| \leq 2|V(G)|$) and $A_p^{-1} \ldots A_2^{-1} B_2 \ldots B_p = \varepsilon$.

It remains to show that there is an algorithm that decides in time $O(|V(G)|\cdot|E(G)|)$ which case holds and outputs $Q$ in Case~\ref{itm:sym:onlyQ}. The algorithm repeats the following for each vertex $v \in V'$:
\begin{itemize}
	\item Compute $\alpha(W_v)$ and $\beta(W_v)$ and search for non
		symmetric edges. If all edges in $\alpha(W_v)$ and $\beta(W_v)$
		are symmetric, continue with the next vertex. Else do the
		next steps.
	\item Check that $\alpha(W_v)$ and $\beta(W_v)$ have the same
		non-symmetric edges appearing in the same order. If not, there is no symmetric walk $Q$ that generates a symmetric walk $S_v$ and we may conclude that Case~\ref{itm:sym:empty} holds.
		Otherwise, let $e_1$ and $e_2$ be respectively the first and
		the last non-symmetric edge of $\alpha(W_v)$ and $\beta(W_v)$.
		Decompose $\alpha(W_v) = A_1 e_1 A_2 e_2 A_3$ and $\beta(W_v) =
		B_1 e_1 B_2 e_2 B_3$.
	\item If in some previous iteration a reduced walk $Q$ has  been fixed,
		check that $Q = A_1B_1^{-1}$. If not, then report Case~\ref{itm:sym:empty}. If no walk $Q$ has been fixed in an earlier iteration, fix $Q$ to be the reduced walk $A_1B_1^{-1} \in \pi(H)$.
	\item Finally, check that  $A_p^{-1} \ldots A_2^{-1} B_2 \ldots B_p = \varepsilon$. If this equality doesn't hold, report Case~\ref{itm:sym:empty}.
\end{itemize}

Then, report case~\ref{itm:sym:onlyQ} and $Q$ if some walk $Q$ has been fixed. Otherwise, report case~\ref{itm:sym:all}.

Each step runs in time $O(|E(G)|)$ and will repeat at most $|V(G)|$ times, so this algorithm runs in time $O(|V(G)| \cdot |E(G)|)$.
\end{proof}

\section{Loopless digraphs} 
\label{sec:monochromatic}

In this section we assume that $G$ and $H$ are loopless weakly connected digraphs.
We extend the polynomial-time algorithm from~\cite{Wrochna:20} for \recol{H} for the case that $H$ is symmetric and square-free to the case where $H$ is any digraph that contains no 4-cycle of algebraic girth 0.
Observe that since $H \subseteq \bar H$, any $H$-recoloring sequence is also a $\bar H$-recoloring sequence. The main idea is to apply Wrochna's algorithm on the symmetric graphs $\bar G$ and $\bar H$ and then use a classification of the possible $\bar H$-recoloring sequences to check whether one of these sequences is compatible with $H$. 
The motivation for excluding 4-cycles of algebraic girth 0 as subgraphs of $H$ is that, for such digraphs $H$, any $H$-recoloring sequence satisfies the monochromatic neighborhood property. 
To see this, consider a step of any $H$-recoloring sequence where the color a vertex $u \in V(G)$ changes from $a$ to $b$ ($a, b \in V(H)$). Let $v$ be any neighbor of $u$ and let $h$ be the current color of $v$. If a color different from $h$ appears in the neighborhood of $u$ then $H$ contains a cycle of algebraic girth 0, that is, one of the two orientations of the $4$-cycle shown in Figure~\ref{fig:zerogirth}.

We say that a \undir{H}-recoloring sequence is \emph{orientation compatible} if
it induces a $H$-recoloring sequence. That is, the homomorphisms of the
sequence are compatible with the orientation of the arcs of $G$ and $H$. For
$\alpha,\beta \colon G \to H$ and $q \in V(G)$, we denote by $\Pi_q$ the set of
walks from $\alpha(q)$ to $\beta(q)$ that are $H$-realizable.  We will
introduce a simple condition, the \emph{zigzag} condition, which allows us to
describe the set of $H$-realizable walks as follows.

\begin{theorem} 
	\label{thm:c4:classification}
	Let $\alpha,\beta \colon G \to H$. We can find in time $O(|V(G)|)$ a vertex $q \in V(G)$ such that one of the following holds:
	\begin{enumerate}
		\item $\Pi_q = \emptyset$. \label{itm:main1:empty}
		\item $\Pi_q = \lbrace Q \rbrace$ for some $Q \in \pi(H)$. \label{itm:main1:Q}
		\item $\Pi_q = \lbrace R^nP \rbrace$ for some $R,P \in \pi(H)$. \label{itm:main1:RnP}
		\item $\Pi_q$ contains all reduced walks of even length from $\alpha(q)$ to $\beta(q)$ that satisfy the zigzag condition. \label{itm:main1:all}
	\end{enumerate}
	Moreover, we can determine in time $O(|V(G)| \cdot |E(G)| + |E(H)|)$ which case holds and output $Q$ or $R,P$ in cases~\ref{itm:main1:Q} or~\ref{itm:main1:RnP}.
\end{theorem}

Furthermore, we show that, given an $H$-realizable walk $Q$, a corresponding
$H$-recoloring sequence can be found in polynomial time in the size of $G$ and
$Q$ (see Lemma~\ref{lem:consdir}). Combining Lemma~\ref{lem:consdir} with
Theorem~\ref{thm:c4:classification} yields the following corollary, which
immediately  Theorem~\ref{thm:c4free}.
\begin{corollary} 
	For any simple digraph $H$ without $4$-cycle of algebraic girth $0$, the problem \recol{H} can be solved in time $O(|E(G)| \cdot |V(G)| + |E(H)|)$.
\end{corollary}

In the next subsection we will recall several tools from~\cite{Wrochna:20}.
Among them is a description of \undir{H}-realizable walks that is analogous to
Theorem~\ref{thm:c4:classification} for undirected graphs. We will then
formally introduce the zigzag condition and prove
Theorem~\ref{thm:c4:classification} in Section~\ref{sec:orient}.

\subsection{Realizable Walks in $\undir{H}$}

In this section we recall some results from~\cite{Wrochna:20} that we will need later on. We will state these results in a slightly different manner for the sake of better integration in our more general setting of digraph homomorphisms. Largely the same proofs apply however, as pointed out in Remark~\ref{rem:nostructure}.

Let $S = \sigma_0, \ldots, \sigma_\ell$ be a $\undir{H}$-reconfiguration
sequence satisfying the monochromatic neighborhood property and let $v \in
V(G)$. For each $0 \leq i < \ell$, let $S_i(v)$ be given by
\[
	S_i(v) =
	\begin{cases}
		\epsilon	&	\text{if $\sigma_i(v) = \sigma_{i+1}(v)$, and}\\
		(\sigma_i(v)\,h)(h\, \sigma_{i+1}(v))	& \text{otherwise,}
	\end{cases}
\]
where $h$ be the unique color of the neighbors of $v$ with respect to
$\sigma_i(v)$ and $\sigma_{i+1}(v)$.
We associate with $S$ and $v$ the walk $S(v) := S_0(v)S_1(v) \cdots
S_\ell(v)$.  Suppose that $S(v) =
(a_1\,a_2)(a_2\,a_3)\cdots(a_{n-2}\,a_{n-1})(a_{n-1}\,a_n)$. Then according to
$S$ the vertex $v$ changes its color from $a_1$ to $a_3$ while its neighbors
have color $a_2$, then it changes color from $a_3$ to $a_5$ while its neighbors
all have color $a_4$ (so all the neighbors must change their color from $a_2$
to $a_4$ before), and so on until $v$ changes its color from $a_{n-2}$ to $a_n$
while its neighbors all have color $a_{n-1}$.

In~\cite{Wrochna:20}, Wrochna showed that the monochromatic neighborhood
property implies that for any two vertices $u, v \in V(G)$ the reduction of
$S(v)$ arises from $S(u)$ by conjucation as follows. By combining~\cite[Lemma
4.1]{Wrochna:20} and Remark~\ref{rem:nostructure} the following lemma is
immediate.
\begin{lemma}[\cite{Wrochna:20}]
	\label{lem:gen}
	Let $S = \sigma_0, \ldots, \sigma_\ell$ be an $\undir{H}$-recoloring
	sequence from $\alpha = \sigma_0$ to $\beta = \sigma_\ell$ satisfying
	the monochromatic neighborhood property and let $W$ be any walk in $G$
	connecting two vertices $u$ and $v$. Then $S(v) = \alpha(W)^{-1} \cdot
	S(u) \cdot \beta(W)$ in $\pi(H)$.  
\end{lemma}

If we pick any vertex $q \in V(G)$, we can observe (following Wrochna) that
$S(q)$ must be topologically valid. We will say that a reduced walk $Q$ from
$\alpha(q)$ to $\beta(q)$ is $\undir{H}$-\emph{realizable} for $\alpha, \beta,
q$ if there is a $\undir{H}$-recoloring sequence $S$ satisfying the monochromatic
neighborhood property such that $Q = S(q)$. We say that $Q$ is
$H$-\emph{realizable} for $\alpha, \beta, q$ if there is a $H$-recoloring
sequence satisfying the monochromatic neighborhood property $S$ such that $Q =
S(q)$. Clearly, if $Q$ is $H$-realizable then it is $\undir{H}$-realizable\footnote{When $\undir{H}$ is square-free, this definition of $H$-realizabiity is exactly the same that the one introduced in Wrochna's paper. So this definition generalizes Wrochna's such that $\undir{H}$-realizability becomes a necessary condition for $H$-realizability.}.

The main result from~\cite{Wrochna:20} we will use is the following
classification of $H$-realizable walks, which can be exploited algorithmically.

\begin{theorem}[Wrochna's classification, \cite{Wrochna:20}] 
	\label{thm:AlWro} 
	Let $\alpha, \beta \colon G \to H$ and $q \in V(G)$. Let $\undir{\Pi}$
	be the set of all reduced walks that are $\undir{H}$-realizable for
	$\alpha, \beta, q$. One of the following holds:
	\begin{enumerate}
		\item $\undir{\Pi} = \emptyset$. \label{itm:nowalk}
		\item $\undir{\Pi} = \lbrace Q \rbrace$ for some $Q \in \pi(H)$. \label{itm:onlyQ}
		\item $\undir{\Pi} = \lbrace R^n P \mid n \in \mathbb{Z} \rbrace$, for some $R,P \in \pi(H)$. \label{itm:walksRnP}
		\item $\undir{\Pi}$ contains all reduced walks of even length from $\alpha(q)$ to $\beta(q)$. \label{itm:allwalks}
	\end{enumerate}
	Furthermore, there is an algorithm that determines in time
	$O(|V(G)|\cdot|E(G)| + |E(H)|)$ which case holds and outputs $Q$ or
	$R,P$ in cases~\ref{itm:onlyQ}~\ref{itm:walksRnP} such that
	$|Q|,|R|,|P|$ are bounded by the total running time
	$O(|V(G)|\cdot|E(G)| + |E(H)|)$.
\end{theorem}

In order to check whether a given $\undir{H}$-realizable walk corresponds to a
$\undir{H}$-recoloring sequence that is compatible with the orientation of the
edges of $H$ we need some more ingredients from~\cite{Wrochna:20}.  There,
Wrochna showed that a $\undir{H}$-recoloring sequence satisfying the
monochromatic neighborhood property can be computed from an
$\undir{H}$-realizable walk in polynomial time. Furthermore, there is a
polynomial-time algorithm that decides whether a given walk from $\alpha(q)$ to
$\beta(q)$ is $\undir{H}$-realizable.
\begin{lemma}[{\cite[see Theorem~6.1]{Wrochna:20}}] 
	\label{thm:conswro} 
	Given an $\undir{H}$-realizable walk $Q$ we can construct an associated $\undir{H}$-recoloring sequence in time $O(|V(G)|^2 + |V(G)|\cdot|Q|)$.
\end{lemma}
Wrochna's proof of the previous construction algorithm contains the following lemma:
\begin{lemma}[\cite{Wrochna:20}] 
	\label{lem:testwro}
	Given any walk $Q$ from $\alpha(q)$ to $\beta(q)$, we can decide in
	time $O(|E(G)| \cdot (|Q| + |V(G)|))$ if $Q$ is $\undir{H}$-realizable.
\end{lemma}

\subsection{Orientation compatibility}
\label{sec:orient}

In this section we characterize the $\undir{H}$-recoloring sequences that satisfy the monochromatic neighborhood property and that are orientation compatible. To
this end we give a simple condition, the \emph{zigzag condition}, on the
reduced walks of the vertices of $G$ obtained from a given $H$-recoloring
sequence.  We show that it suffices to check the zigzag condition for all vertices of $G$ in order to check whether a \undir{H}-recoloring
sequence is orientation compatible. We use this insight to obtain a
polynomial-time algorithm that, given two $H$-colorings $\alpha, \beta : G \to
H$ and a vertex $v$ of $G$, finds a walk from $\alpha(v)$ to $\beta(v)$ in $H$
that is compatible with the orientation of $H$ or reports correctly that no
such walk exists. 
In the following, $G,H$ are any simple, (weakly) connected digraph with an arc and $\alpha,\beta \colon G \to H$.

\subsubsection{The zigzag condition}

Let $v \in V(G)$ and $S_v = (a_1 a_2) \ldots (a_{n-1} a_n)$ be a walk of even length from $\alpha(v)$ to $\beta(v)$ (think of $S_v$ as the walk of $v$ under a $\undir{H}$-reconfiguration $S$ that satisfies the monochromatic neighborhood property; in particular, $S$ and $S_v$ could have been obtained by combining Theorem~\ref{thm:AlWro} and Lemma~\ref{thm:conswro}).
We say that $S_v$ satisfies the \emph{zigzag condition} if $a_1 \leftarrow a_2
\rightarrow a_3 \leftarrow \ldots \leftarrow a_{n-1} \rightarrow a_n$ is a path
in $H$ whenever $N^-_G(v) \neq \emptyset$ and $a_1 \rightarrow a_2 \leftarrow
a_3 \rightarrow \ldots \rightarrow a_{n-1} \leftarrow a_n$ is a path of $H$
whenever $N^+_G(v) \neq \emptyset$. Note that at least one of $N^-_G(v)$ and $N^+_G(v)$ is non-empty since $G$ is weakly connected. Also observe that if $S_v$ is not reduced and satisfies the zigzag condition, then it will still satisfy that condition after reduction.

Using Lemma~\ref{lem:gen}, we can state a more general and useful definition of
orientation compatibility for walks: given $\alpha, \beta \colon G \to H$, $q
\in V(G)$ and a system of walks $(W_v)$ from $q$ to $v$. We say that a reduced
walk of even length $Q \in \pi(H)$ from $\alpha(q)$ to $\beta(q)$ is
\emph{orientation compatible} for the system $(W_v)_{v \in V(G)}$ if for each
vertex $v \in V(G)$ and any walk $W_v$ from $q$ to $v$, the walk
$\alpha(W_v)^{-1} \cdot Q \cdot \beta(W_v) \in \pi(H)$ satisfies the zigzag condition. See figure~\ref{fig:zigzag} for illustrations of the zigzag condition and of orientation compatibility.
Observe that by Lemma~\ref{lem:equigen}, if $Q$ is topologically valid (in particular, if $Q$ is $\undir{H}$-realizable), then the orientation compatibility of $Q$ does not
depend of the choice of the walks $(W_v)_{v \in V(G)}$.

\begin{figure}
\begin{tikzpicture}[scale=0.7]
    \begin{scope}
        \dstardiamgraph;
        \node[G,label=72:$v_0$] (v0) at ($(ha1)$){};
        \node[G,label=-144:$v_1$] (v1) at ($(ha3)$){};
        \node[G,label=-72:$v_2$] (v2) at ($(ha4)$){};
        \node[G] (v3) at ($(ha0)$){};
        \node[G,label=90:$v_3$] (v4) at ($(hb0)$){};
        \node[G,label=-120:$q$] (v5) at ($(hb3)$){};
        \draw[Ge] (v0)--(v1)--(v2)--(v0)--(v3)--(v4)--(v5);
        \node at (3,-1.3) {\Large$\alpha$};
    \end{scope}
    \begin{scope}[shift={(12,0)}]
        \dstardiamgraph;
        \node[G,label=72:$v_0$] (v0) at ($(ha1)$){};
        \node[G,label=-144:$v_1$] (v1) at ($(ha3) - (0:6pt)$){};
        \node[G,label=-144:$v_2$] (v2) at ($(ha4) - (0:6pt)$){};
        \node[G] (v3) at ($(ha4) + (0:6pt)$){};
        \node[G,label=72:$v_3$] (v4) at ($(ha0)$){};
        \node[G,label=72:$q$] (v5) at ($(ha3) + (0:6pt)$){};
        \draw[Ge] (v0)--(v1)--(v2)--(v0)--(v3)--(v4)--(v5);
        \node at (3,-1.3) {\Large$\beta$};
    \end{scope}
\end{tikzpicture}

\begin{tikzpicture}[scale=0.7]
    \begin{scope}
        \dstardiamgraph;
        \node[G1] (v0) at ($(hb0)$){};
        \node[G] (v1) at ($(ha0) + (90:5pt)$){};
        \node[G] (v2) at ($(ha1)$){};
        \node[G] (v3) at ($(ha4)$){};
        \node[G1] (v4) at ($(ha0) - (90:5pt)$){};
        \draw[Ge] (v0)--(v1) node[sloped,pos=0.55,allow upside down]{\arrowInG}; ;
        \draw[Ge] (v1)--(v2) node[sloped,pos=0.55,allow upside down]{\arrowInG}; ;
        \draw[Ge] (v2)--(v3) node[sloped,pos=0.55,allow upside down]{\arrowInG}; ;
        \draw[Ge] (v3)--(v4) node[sloped,pos=0.55,allow upside down]{\arrowInG}; ;
        \node at (3,-1.3) {\Large$S(v_3)$};
    \end{scope}
    \begin{scope}[shift={(12,0)}]
        \dstardiamgraph;
        \node[G1] (v0) at ($(hb3)$){};
        \node[G] (v1) at ($(hb0)$){};
        \node[G] (v2) at ($(ha0) + (90:5pt)$){};
        \node[G] (v3) at ($(ha1)$){};
        \node[G] (v4) at ($(ha4)$){};
        \node[G] (v5) at ($(ha0) - (90:5pt)$){};
        \node[G1] (v6) at ($(ha3)$){};
        \draw[Ge] (v0)--(v1) node[sloped,pos=0.55,allow upside down]{\arrowInG}; ;
        \draw[Ge] (v1)--(v2) node[sloped,pos=0.55,allow upside down]{\arrowInG}; ;
        \draw[Ge] (v2)--(v3) node[sloped,pos=0.55,allow upside down]{\arrowInG}; ;
        \draw[Ge] (v3)--(v4) node[sloped,pos=0.55,allow upside down]{\arrowInG}; ;
        \draw[Ge] (v4)--(v5) node[sloped,pos=0.55,allow upside down]{\arrowInG}; ;
        \draw[Ge] (v5)--(v6) node[sloped,pos=0.55,allow upside down]{\arrowInG}; ;
        \node at (3,-1.3) {\Large$Q=S(q)$};
    \end{scope}
\end{tikzpicture}

\caption{The walk $Q=S(q)$ satisfies the zigzag condition and generates walks that also satisfy the zigzag condition like $S(v_3)$, so it is orientation compatible. We will see soon that the vertices of the triangle $v_0 \, v_1 \, v_2$ must have symmetric vertex walks, so it follows that they are frozen (since there is no symmetric edge in $H$ here).}
\label{fig:zigzag}
\end{figure}
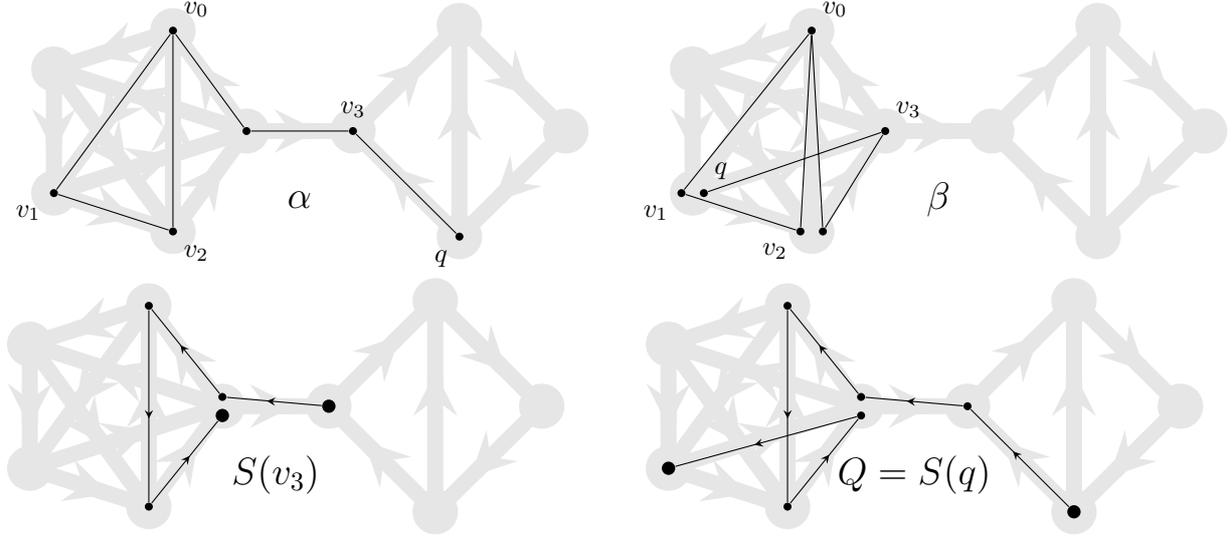

\begin{lemma}
	\label{lem:roc} 
	Let $\alpha, \beta \colon G \to H$ and $q \in V(G)$. Let $Q$ be an $\undir{H}$-realizable walk for $\alpha, \beta, q$. Then $Q$ is $H$-realizable if and only if it is orientation compatible.
\end{lemma}

\begin{proof} 
        Let $Q$ be an $\undir{H}$-realizable walk. First assume that $Q$ is $H$-realizable. Then let $S$ be a $H$-recoloring sequence from $\alpha$ to $\beta$ such that $Q = S(q)$. Let $v \in V(G)$ and $W_v$ a walk from $q$ to $v$ in $G$. Then $\alpha(W_v)^{-1} Q \beta(W_v) = S(v)$ and we can write $S(v)$ as a reduced walk $(a_1 a_2) (a_2 a_3) \ldots (a_{n-1} a_n)$. If $N^-_G(v) \neq \emptyset$ then there is an arc $w \rightarrow v$ for some $w \in V(G)$. At each color change of $v$ from $a_i$ to $a_{i+2}$, the vertex $w$ must have color $a_{i+1}$ because $S$ satisfies the monochromatric neighborhood property. Since $S$ is a $H$-recoloring sequence, its homomorphisms all preserve the arc $w \rightarrow v$.  Hence $a_1 \leftarrow a_2
        \rightarrow \ldots \leftarrow a_{n-1} \rightarrow a_n$ is a path in $H$.  Similarly, if $N^+_G(v) \neq \emptyset$ then we have the same path with all arcs reversed. As this holds for each vertex $v \in V(G)$, $Q$ is orientation compatible.\\
        Conversely, if $Q$ is orientation compatible, then let $S$ be the $\undir{H}$-recoloring sequence constructed via Lemma~\ref{thm:conswro}. Consider any step $\sigma_i, \sigma_{i+1}$ of $S$, say when a vertex $v$ changes color from $a$ to $b$ while its neighbor have color $c$. This color change is recorded in $S(u)$ which satisfies the zigzag condition so if $\sigma_i$ induces a homomorphism $G \to H$, then $\sigma_{i+1}$ too. By induction, each homomorphism of $S$ induces a homomorphism $G \to H$ so we eventually have a $H$-recoloring sequence.
\end{proof}

Due to the second part of the proof of Lemma~\ref{lem:roc}, the following generalization of 
Lemma~\ref{thm:conswro} is immediate.

\begin{lemma}
	\label{lem:consdir} 
	Given an $H$-realizable walk $Q$, we can construct an associated
	$H$-recoloring sequence in time $O(|V(G)|^2 + |V(G)|\cdot |Q|)$.
\end{lemma}

Furthermore, the zigzag condition can be exploited in order to decide
efficiently whether a given walk from $\alpha(q)$ to $\beta(q)$ is
$H$-realizable.

\begin{lemma} 
	\label{lem:test} 
	There is a polynomial-time algorithm that, given a walk $Q$ from $\alpha(q)$ to $\beta(q)$, decides in time $O(|E(G)| \cdot (|Q| + |V(G)|))$ if $Q$ is $H$-realizable for $\alpha, \beta, q$.
\end{lemma}
\begin{proof} 
	Use Lemma~\ref{lem:testwro} to decide in time $O(|E(G)| \cdot (|Q| +
	|V(G)|))$ whether $Q$ is $\undir{H}$-realizable for $\alpha,\beta,q$. 
	If not then $Q$ is not $H$-realizable for $\alpha, \beta, q$. Otherwise, 
	for each vertex $v \in V(G)$, use breadth first search to find a
	shortest walk $W_v$ from $q$ to $v$ in $G$ in time $O(|E(G)|)$, then
	compute $S(v) := \alpha(W_v)^{-1} \cdot Q \cdot \beta(W_v)$, reduce the
	resulting walk and check the zigzag condition in time $O(|Q| +
	|V(G)|)$. By Lemma~\ref{lem:roc}, the zigzag condition is satisfied for
	each vertex $v \in V(G)$ if and only if $Q$ is $H$-realizable for
	$\alpha, \beta, q$. 
	In total, for each vertex $v \in V(G)$ the computations can be
	performed in time $O(|Q| + |E(G)|)$. 
\end{proof}

\subsubsection{A description of orientation-compatible walks}

Let $Q$ be an even walk from $\alpha(q)$ to $\beta(q)$. Furthermore, let $v \in
V(G)$ and let $W_v$ be a walk from $q$ to $v$. Suppose that $S_v := \alpha(W_v)
\cdot Q \cdot \beta(W_v) = (a_1\,a_2)\ldots(a_{n-1}\,a_n)$ is the reduced walk
associated to $v$ by $Q$ and $W_v$. We say that $v$ is of type \IN if $N^-_G(v)
\neq \emptyset$ and it is of type \OUT if $N^+_G(v) \neq \emptyset$.
Furthermore, we say that $S_v$ is \emph{\IN-compatible} (resp.,
\emph{\OUT-compatible}) if $v$ is of type \IN and additionally $a_1 \leftarrow
a_2 \rightarrow a_3 \leftarrow \ldots \leftarrow a_{n-1} \rightarrow a_n$ is a
path in $H$ (resp., $v$ is of type \OUT  and additionally $a_1 \rightarrow a_2
\leftarrow a_3 \rightarrow \ldots \rightarrow a_{n-1} \leftarrow a_n$ is a path
of $H$). Finally, if $v$ is of type \IN and of type \OUT, we say that $v$ is of
type \SYM and that $S_v$ is \SYM-compatible if it is \IN-compatible and
\OUT-compatible, this means in particular that $S_v$ has only symmetric edges.

By Lemma~\ref{lem:roc}, we have that an even walk $Q$ is orientation compatible
for the system $(W_v)_{v \in V(G)}$ if and only if for every vertex $v$, if $v$
is of type \IN, then $S_v$ is \IN-compatible and if $v$ is of type \OUT, then
$S_v$ is \OUT-compatible.

\begin{lemma} 
	\label{lem:inout}
	For any arc $u \rightarrow v$ of $G$, the walk $S_u$ is \OUT-compatible if and only if $S_v$ is \IN-compatible.
\end{lemma}
\begin{proof} 
	Note that $u$ is of type \OUT and $v$ is of type \IN. By the
	monochromatic neighborhood property, if $S_u$ is \OUT-compatible then since $\alpha$ and $\beta$ are homomorphisms, $(\alpha(v)\, \alpha(u)) S_u (\beta(u)\, \beta(v))$, is \IN-compatible. As $S_v$ is precisely this walk after reduction, it is then \IN-compatible. Similarly, if $S_v$ is \IN-compatible, then $S_u$ is \OUT-compatible.
\end{proof}

\begin{lemma}
    \label{lem:compSYM}
    \mbox{}
	\begin{enumerate}
		\item If $G$ has no vertex of type \SYM and $Q$ satisfies the zigzag condition then $S_v$ satisfies the zigzag condition for all $v \in V(G)$.
		\item Let $\{v_1, \ldots, v_k\} \subseteq V(G)$ be the subset of vertices of $G$ of type \SYM. If $S_{v_i}$ satisfies the zigzag condition for $1 \leq i \leq k$ then $S_v$ satisfies the zigzag condition for all $v \in V(G)$.
	\end{enumerate}
\end{lemma}

\begin{proof}
	We prove the first statement.  Assume there is no vertex of type \SYM
	and that $Q$ satisfies the zigzag condition. Without loss of
	generality, we may assume that $q$ is of type \IN, so $Q$ is
	\IN-compatible. Let $v$ be any other vertex of $G$ and $P$ a path from
	$q$ to $v$ in $G$. Since there is no vertex of type \SYM we deduce that
	$P$ is alternating between vertices of type \IN and vertices of type
	\OUT. By Lemma~\ref{lem:inout} $S_w$ satisfies the zigzag condition for
	each vertex $w$ in $P$. In particular, $S_w$ does.
	\label{itm:zigzag:nosym}
	
	It remains to prove the second statement.
	Let $X = \{v_1, \ldots, v_k\} \subseteq V(G)$ be the vertices of $G$ of
	type \SYM and suppose that $S_{v_i}$ satisfies the zigzag condition for
	$1 \leq i \leq k$.  Let $v \in V(G)$ be any vertex that is not of type
	\SYM and let $P$ be a shortest path from $X$ to $v$. Again, $P$ is
	alternating between vertices of type \IN and vertices of type \OUT and
	hence for each vertex $w$ of $P$, we obtain that $S_w$ satisfies the
	zigzag condition by Lemma~\ref{lem:inout}. \label{itm:zigzag:sym}
\end{proof}

We obtain a description of all orientation compatible walks. 

\begin{theorem} 
	\label{thm:oriwalks} 
	Let $(W_v)$ be a system of walks from $q_0$ to all vertices $v \in
	V(G)$. There is some vertex $q \in V(G)$ such that the set of
	orientation compatible walks for $q$ and the system $(W_q^{-1}W_v)_{v
	\in V(G)}$ is one of the followings:
	\begin{enumerate}
		\item $\emptyset$. \label{itm:ori:nowalk}
		\item $\lbrace Q \rbrace$ for some reduced walk $Q$ of even length $|Q| = O(|V(G)|)$. \label{itm:ori:Q}
		\item The set of all reduced walks of even length from $\alpha(q)$ to $\beta(q)$ that satisfy the zigzag condition. \label{itm:ori:all}
	\end{enumerate}
	Furthermore, we can determine in time $O(|V(G)| \cdot |E(G)|)$ which
	case holds, and output $Q$ in Case~\ref{itm:ori:Q}.
\end{theorem}
\begin{proof} 
	Let $V' \subseteq V(G)$ be the set of vertices of type \SYM. Notice
	that $V'$ can be computed in time $O(|E(G)|)$.
	Suppose first, suppose that $V' = \emptyset$. Then, by
	statement~\ref{itm:zigzag:nosym} the orientation-compatible walks are
	precisely those of even lenth that satisfy the zigzag-condition. We can
	therefore indicate Case~\ref{itm:ori:all} with $q := q_0$. Now suppose
	that $V' \neq \emptyset$. Let $q \in V'$ and apply
	Lemma~\ref{lem:symwalks} to determine in time $O(|V(G)| \cdot |E(G)|)$
	all walks from $\alpha(q)$ to $\beta(q)$ that generate symmetric walks
	for the system $(W_v)_{v \in V'}$ on all vertices of $V'$. Invoke the
	second statement of Lemma~\ref{lem:compSYM} to deduce that:
	\begin{enumerate}
		\item In Case~\ref{itm:sym:empty} of Lemma~\ref{lem:symwalks} we report Case~\ref{itm:ori:nowalk}, i.e., there is no orientation-compatible walk.
		\item In Case~\ref{itm:sym:onlyQ} of Lemma~\ref{lem:symwalks} we report Case~\ref{itm:ori:Q} and output $Q$ if $Q$  has even length and Case~\ref{itm:ori:nowalk} otherwise.
		\item In Case~\ref{itm:sym:all} of Lemma~\ref{lem:symwalks} we report Case~\ref{itm:ori:all}.
	\end{enumerate}
	The total runtime is dominated by the algorithm of Lemma~\ref{lem:symwalks}, hence $O(|V(G)|\cdot|E(G)|)$ as claimed.
\end{proof}

\subsection{Proof of Theorem~\ref{thm:c4:classification}}

\begin{lemma} 
	\label{lem:Rcancel}
    Let $R_0$ be cyclically reduced walk and let $P$ be a reduced walk starting at the base point of $R_0$.  We can find in time $O(|R_0| + |P|)$ an
	integer $n_0$ such that none of $R_0$ and $R_0^{-1}$ entirely cancels
	with $R^{n_0}P$.
\end{lemma}

\begin{proof}
	Start with $n=0$. Check if $R_0$ entirely reduces with $R_0^nP$ and if so then replace $n$ by $n+1$. Similarly, if $R_0^{-1}$ reduces with $R_0^nP$, then replace $n$ by $n-1$. Repeat until none of $R_0^{-1}$ or $R_0$ entirely reduces with $R_0^nP$. Each step is done in $|R_0|$ and reduces $|R_0^nP|$ by $|R_0|$, so we deduce that this process terminate in time $O(|R_0| + |P|)$. Also observe that $|n_0|$ is polynomial in $|V(G)|$ and $|V(H)|$.
\end{proof}

Recall that for $\alpha,\beta \colon G \to H$ and $q \in V(G)$, we denote by
$\Pi_q$ the set of all walks from $\alpha(q)$ to $\beta(q)$ that are
$H$-realizable. We are now ready to prove Theorem~\ref{thm:c4:classification}.

\begin{proof}[Proof of Theorem~\ref{thm:c4:classification}] 
	Fix any $q_0 \in V(G)$ and use breadth first search to compute shortest
	walks $W_v$ from $q_0$ to $v$ for all $v \in V(G)$ in time $O(|V(G)|
	\cdot |E(G)|)$.  We invoke Theorem~\ref{thm:oriwalks} for $q_0$ and
	$(W_v)_{v \in V}$ to obtain in time $O(|V(G)| \cdot |E(G)|)$ a vertex
	$q \in V(G)$ and a description of the set of orientation compatible
	walks for $q$ and $(W_v)_{v \in V(G)}$. We distinguish the possible
	outcomes: 
	\begin{description}
		\item[Case 1] There is no orientation-compatible walk. Then report that there is no $H$-realizable walk.
		\item[Case 2] There ia a unique orientation-compatible reduced walk $Q$ of even length. By Lemma~\ref{lem:test}, we can decide in time $O(|E(G)| \cdot (|Q| + |V(G)|) = O(|V(G)| \cdot |E(G)|)$ if $|Q|$ is $H$-realizable as $|Q| = O(|V(G)|)$.
		\item[Case 3] All reduced walks of even length from $\alpha(q)$
			to $\beta(q)$ that satisfy the zigzag condition and are
			orientation compatible. Invoke Theorem~\ref{thm:AlWro}
			to get in time $O(|V(G)|\cdot |E(G)| + |E(H)|)$ a description of the set $\undir{\Pi}_q$ of all
			$\undir{H}$-realizable walks from $\alpha(q)$ to
			$\beta(q)$. We again distinguish the four possible outcomes.
			\begin{enumerate}
				\item $\undir\Pi_q = \emptyset$. There is no $\undir{H}$-realizable walk for $\alpha,\beta,q$, so there is no $H$-realizable walk.
				\item $\undir\Pi_q = \{Q \}$. There is a unique reduced walk $Q$ that is $\undir{H}$-realizable. Then we can check in time $O(|Q|) = O(|E(G)| \cdot |V(G)| + |E(H)|)$ whether it satisfies the zigzag condition and hence is orientation compatible. 
				\item $\undir{\Pi}_q = \lbrace R^nP \mid n \in \mathbb{Z} \rbrace$ with $R,P \in \pi(H)$ and $R$ closed and of even length. If $R$ does not satisfy the zigzag condition then the following claim allows us to conclude.
					\begin{myclaim}
						Suppose that $R$ does not satisfy the zigzag condition. Then at most one of the $\undir{H}$-realizable walks $R^nP$, $n \in \mathbb{Z}$, satisfies the zigzag condition. Furthermore, we can find such a walk in time $O(|V(G)| \cdot |E(G)| + |E(H)|)$ or conclude there is none.
					\end{myclaim}
					\begin{proof}
						Decompose $R = A R_0 A^{-1}$ with all walks minimal and $R_0$ cyclically reduced. Apply Lemma~\ref{lem:Rcancel} with the walks $R_0$ and $A^{-1}P$ and obtain in time $O(|R_0| + |P|) = O(|V(G)| \cdot |E(G)| + |E(H)|)$ an integer $n_0 \in \mathbb{Z}$ such that none of $R_0^{-1}$ and $R_0$ entirely cancels with $R_0^{n_0} A^{-1}P$.
						
						If $A$ does not satisfy the zigzag condition, then none of $R^nP$ do for $n \neq n_0$, so only $R^{n_0}P$ can possibly satisfy the zigzag condition, which can be checked in time $O(|R^{n_0}P|) = O(|V(G)| \cdot |E(G)| + |E(H)|)$.

						Otherwise, if $A$ satisfies the zigzag condition, then so do the edges of $A^{-1}$, so $R_0$ does not satisfy the zigzag condition (since $R$ does not). For $n > n_0 +1$, $R^nP \in \pi(H)$ contains an entire $R_0$ that does not reduce, so it cannot satisfy the zigzag condition. Similarly if $n < n_0 -1$, then $R^nP$ does not satisfy the zigzag condition since it contains an entire $R_0^{-1}$. Eventually we only need to test $R^{n_0-1}P$, $R^{n_0}P$ and $R^{n_0+1}P$, which can be done in time $O(|P|) = O(|V(G)| \cdot |E(G)| + |E(H)|)$. Observe that each edge of $R_0$ belongs to precisely two of those three walks, so as it is the case for the edges of $R_0$ that do not fit the zigzag condition, we deduce that at most one of $R^{n_0-1}P$, $R^{n_0}P$ and $R^{n_0+1}P$ satisfy the zigzag condition.
					\end{proof}
					On the other hand, if $R$ satisfies the zigzag condition then $P$ satisfies the zigzag condition if and only all walks $R^nP$ do. To see this, notice that  ''badly oriented'' edges of $P$ must reduce with ''badly oriented'' edges of $R^n$, but there is none in $R^n$ since it is orientation-compatible. So we can again distinguish between Case~\ref{itm:main1:empty} and Case~\ref{itm:main1:RnP} in time $O(|E(G)| \cdot |V(G)| + |E(H)|)$ and report the result.
				\item $\undir\Pi_q$ contains all reduced walks of even length from $\alpha(q)$ to $\beta(q)$. We report Case~\ref{itm:main1:all}.
			\end{enumerate}
	\end{description}
\end{proof}

\section{Reflexive graphs}
\label{sec:push-pull}

In this section we prove Theorem~\ref{thm:reflexive}. To this end we show that there is a polynomial-time algorithm for \recolun{H} for any reflexive digraph $H$ that neither contains a triangle of algebraic girth 1 nor a 4-cycle of algebraic girth 0. Lemma~\ref{lem:equirecol} then implies that \recol{H} admits such an algorithm as well.
We first consider the case where $H$ is undirected and reprove the following result by Lee et al.~\cite{Lee:21}.

\begin{theorem}[{\cite[Theorem 1.1]{Lee:21}}]
    \label{thm:triangle-free}
Let $H$ be an undirected, reflexive, triangle-free graph. Then \recolun{H} admits a polynomial-time algorithm for reflexive instances.
\end{theorem}

Our proof has the following advantages over the one given in~\cite{Lee:21}. First, it is simpler in the sense that the result follows from an adaptation of the tools introduced by Wrochna in \cite{Wrochna:20} to the reflexive triangle-free case. In particular, we obtain a characterization of all walks that are $H$-realizable for undirected reflexive triangle-free graphs $H$, that is, walks in $H$ that correspond to the color changes of a vertex that are induced by an $H$-recoloring sequence. Based on this characterization we can prove an extension of Theorem~\ref{thm:triangle-free} to directed graphs as follows.

\begin{theorem} 
\label{thm:triangledir}
Let $H$ be a reflexive digraph that does not contain a triangle of algebraic girth $1$. Then \recolun{H} can be solved in polynomial time for reflexive instances.
\end{theorem}

As observed in~\cite{Lee:21} (in the setting of undirected graphs), adding loops does not change the reconfiguration graph. By combining this observation with Lemma~\ref{lem:equirecol}, we obtain the following result.

\begin{corollary}
Let $H$ be a reflexive digraph that does not contain a triangle of algebraic girth $1$ or a $4$-cycle of algebraic girth $0$. Then \recolun{H} admits a polynomial-time algorithm.
\end{corollary}

\begin{proof} 
Let $(G,\alpha,\beta)$ be an instance of \recolun{H}. Furthermore, let $G'$ be a copy of $G$ with a loop added to each vertex. Observe that the instances $(G, \alpha, \beta)$ and $(G', \alpha, \beta)$ are equivalent for $\recol{H}$. So, by Lemma~\ref{lem:equirecol}, they are equivalent for \recolun{H} too.
\end{proof}

We will focus on solving instances of \recolun{H}. Since $G$ and $H$ are reflexive, this means we consider reconfiguration sequences such that whenever a vertex of $G$ changes its color from $a$ to $b$, then $b$ is a neighbor of $a$.
So any $H$-recoloring sequence $S$ can be associated with a set $\{ S(u) \}_{u \in V(G)}$ of walks in $H$, where $S(u)$ is the walk in $H$ corresponding to the successive colors of $u$ according to $S$.

\subsection{Undirected reflexive graphs}
\label{sec:unreflexive}

In this section, we give an alternative proof of Theorem~\ref{thm:triangle-free} based on Wrochna's algorithm.
As observed in~\cite{Lee:21}, this theorem implies the following result.

\begin{corollary} 
    \label{cor:girth5}
Let $H$ be a reflexive undirected graph of girth at least $5$. Then \recolun{H} can be solved in polynomial time for all instances\footnote{The reader may notice that Lemma~\ref{lem:equirecol} also holds when $H$ is an undirected, reflexive, square-free graph without isolated vertex. So Corollary~\ref{cor:girth5} implies that \recol{H} too can be solved in polynomial time for all instances under those hypothesis on $H$.}.
\end{corollary}

For now we assume that $G$ and $H$ are undirected. It turns out that the proofs given by Wrochna in~\cite{Wrochna:20} only need to be changed slightly in order to show Theorem~\ref{thm:triangle-free}.

We introduce the \pop property, which can be thought of as an analogue of the monochromatic neighborhood property for triangle-free graphs.

\begin{definition}
    Let $\alpha, \beta : G \to H$ and let $S$ be a $H$-recoloring sequence from $\alpha$ to $\beta$. Then $S$ has the \emph{\pop} property if whenever a vertex $u \in V(G)$ changes its color from $a$ to $b$ then any neighbor $v \in V(G)$ of $u$ has color $a$ or $b$.
\end{definition}

Notice that if $H$ is triangle-free then any $H$-recoloring sequence has the \pop property: if a neighbor $v$ of $u$ has a color different from $a$ and $b$ then $H$ contains a triangle.
So in this case, for any $H$-recoloring sequence, we have that the color change of some vertex $u$ must be, given any neighbor of $u$, either a ``push'' from the color of that neighbor or a ``pull'' to the color of that neighbor. 

Let $\alpha, \beta : G \to H$ and let $S$ be an $H$-recoloring sequence from $\alpha$ to $\beta$. For each vertex $v \in V(G)$, let $S(v)$ be the vertex walk in $H$ associated to $v$. Given a vertex $q$ and a walk $Q$ from $\alpha(q)$ to $\beta(q)$. We say that $Q$ is \emph{$H$-realizable} for $\alpha,\beta,q$ if there is an $H$-recoloring sequence $S$ from $\alpha$ to $\beta$ such that $Q = S(q)$ and $S$ satisfies the \pop property.
By the next lemma we have that for any vertex $v \in V(G)$ the corresponding walk $S(v)$ generates the walk $S(u)$ of any other vertex by conjugation. Notice that the same holds for $H$-recoloring sequences that satisfy the monochromatic neighborhood property (\cite[Lemma 4.1]{Wrochna:20}, Lemma~\ref{lem:gen}).
Also observe that the next lemma implies in particular that $H$-realizable walks are topologically valid.

\begin{lemma} 
    \label{lem:fundtri} 
    Let $S$ be a $H$-recoloring sequence from $\alpha$ to $\beta$ satisfying the \pop property. Then for any $u,v \in V(G)$ and any $u$-$v$ walk $W$, we have $S(v) = \alpha(W)^{-1} S(u) \beta(W)$ in $\pi(H)$.
\end{lemma}
\begin{proof}
	We use induction on the length $\ell$ of $S = \sigma_0 \sigma_1 \ldots \sigma_l$. Let $\ell = 1$ and suppose $\sigma_0 \neq \sigma_1$, so a vertex $w \in V(G)$ is recolored
	from $\sigma_0(w) = a$ to $\sigma_1(w) = b$. By definition, we have $S(w) = (ab)$ and $S(v) = \varepsilon$ for $v \in V(G) \setminus \{w\}$. If $W = \varepsilon$
	then $\sigma_0(W) = \sigma_1(W) = \varepsilon$ and $S(u) = S(v)$ since
	$u = v$, so we are done. If $W$ has length one then, without loss of generality, $W = u \to v$. We consider three cases:
	\begin{itemize}
		\item $u \neq w$ and $v \neq w$. Then $S(u) = S(v) = \varepsilon$ and $\sigma_0(W) = \sigma_1(W)$.
		\item $u \neq w$ and $v = w$. Then $S(u) = \varepsilon$, $S(v)
			= (ab)$. Since $S$ satisfies the \pop property, all neighbors of $v$, including $u$, must have color either $a$ or $b$.
			\begin{itemize}
    			\item If $u$ has color $a$, then $\alpha(W) = (aa) = \varepsilon$ and $\beta(W) = (ab)$ so $\alpha(W)^{-1} S(u)\beta(W) = (ab) = S(v)$ in $\pi(H)$. 
    			\item If $u$ has color $b$, then $\alpha(W) = (ba)$ and $\beta(W) = (bb) = \varepsilon$ so $\alpha(W)^{-1} S(u) \beta(W) = (ab) = S(v)$ in $\pi(H)$.
			\end{itemize}
		\item The case when $u = w$ and $v \neq w$ is symmetric. 
	\end{itemize}
	In each case we have $S(v) = \sigma_0(W)^{-1} S(u) \sigma_1(W)$. If
	$W$ has length at least two then we may split $W$ inductively into $W =
	W_1W_2$ such that $W_2$ is of length one, so $W_1$ is a walk from $u$
	to $z$ and $W_2$ is a walk from $z$ to $v$. Then we have
	\begin{align*}
		\sigma_1(W)	&= \sigma_1(W_1)\sigma_1(W_2)  \\
				&= \sigma_1(W_1)S(z)^{-1}\sigma_0(W_2)S(v) \\
				&= S(u)^{-1} \sigma_0(W) S(v) \enspace.
	\end{align*}
	If the sequence $S$ has length more than one we use the same idea and
	split $S$ inductively into $S = S_1S_2$ such that $S_2$ has length one. Then for each $v \in V(G)$ we have $S(v) = S_1(v)S_2(v)$ and
	\begin{align*}
		S(v)	&= S_1(v)S_2(v)  \\
			&= S_1(v)\sigma_{\ell-1}(W)^{-1}S_2(u)\sigma_\ell(W) \\
			&= \sigma_0(W)^{-1}S(u)\sigma_\ell(W) \enspace,
	\end{align*}
	which concludes the proof.
\end{proof}

We will obtain a characterization of $H$-realizable walks (Section~\ref{sec:R:realizable}) that is based on an algorithm that finds vertices of $G$ whose color cannot change (Section~\ref{sec:R:tightclosed}). From this we end up with the following result which immediately implies Theorem~\ref{thm:triangle-free}.

\begin{theorem} \label{thm:R:realwalks} Let $G$ and $H$ be reflexive undirected graphs and let $\alpha, \beta \colon G \to H$ and $q \in V(G)$. Let $\undir{\Pi}$ be the set of all walks that are $H$-realizable for $\alpha, \beta, q$ (in particular, the corresponding $H$-recoloring sequences satisfy the \pop property). Then one of the following holds:
	\begin{enumerate}
		\item $\undir{\Pi} = \emptyset$. \label{itm:R:nowalk}
		\item $\undir{\Pi} = \lbrace Q \rbrace$ for some $Q \in \pi(H)$. \label{itm:R:onlyQ}
		\item $\undir{\Pi}= \lbrace R^n P | n \in \mathbb{Z} \rbrace$, for some $R,P \in \pi(H)$. \label{itm:R:walksRnP}
		\item $\undir{\Pi}$ contains all reduced walks from $\alpha(q)$ to $\beta(q)$. \label{itm:R:allwalks}
	\end{enumerate}
    Furthermore, we can determine in time $O(|V(G)|\cdot|E(G)| + |E(H)|)$ which case holds and output $Q$ or $R,P$ in cases~\ref{itm:R:onlyQ} and ~\ref{itm:R:walksRnP} such that $|Q|,|R|,|P|$ are bounded by the total running time $O(|V(G)|\cdot|E(G)| + |E(H)|)$. Case~\ref{itm:R:allwalks} happens when $\alpha(C) = \beta(C) = \varepsilon$ in $\pi(H)$ for all closed walks $C$ in $G$.
\end{theorem}

\subsubsection{Tight closed walks}
\label{sec:R:tightclosed}

In this section we show how to identify in polynomial time vertices of $G$ whose color cannot change in any $H$-recoloring sequence that satisfies the \pop property.
Let $\alpha, \beta : G \to H$ and let 
$S$ be an $H$-recoloring sequence from $\alpha$ to $\beta$. We say that a vertex $q \in V(G)$ is \emph{frozen} for $S$ if $S(q)$ is the empty walk $\varepsilon$ (in other words, the color of $q$ is never changed by $S$). Recall that a closed walk $C = (v_1\, v_2)\cdots(v_n\,v_1)$ is \emph{cyclically reduced} if $C$ is reduced and additionally $v_n \neq v_2$. We say that a closed walk $C$ in $G$ is $\alpha$-tight if $\alpha(C)$ is cyclically reduced.

\begin{lemma} 
    \label{lem:cons}
    Let $C$ be an $\alpha$-tight closed walk in $G$. Then the vertices of $C$ are frozen for any recoloring sequence $S$ that satisfies the \pop property.
\end{lemma}
\begin{proof} 
    By contraposition, let $v$ be the first vertex of $C$ that changes its color during some step of a $H$-recoloring sequence $S$ satisfying the \pop property. We will show that $C$ is not $\alpha$-tight. Furthermore, assume that the color of $v$ changes from $a$ to $b$ in this step. Let $u,w$ be the two neighbors of $v$ on $C$. Since $v$ is the first vertex of $C$ that changes color, the colors of $u$ and $w$ must be $\alpha(u)$ and $\alpha(w)$, respectively, during this step and $\alpha(u), \alpha(w) \in \{a, b\}$ by the \pop property. If either of $u$ and $w$ has color $a = \alpha(v)$ then there is an edge $(a\,a)$ in $\alpha(C)$, so $\alpha(C)$ is not cyclically reduced. On the other hand, if both $u$ and $w$ have color $b$ then we have consecutive edges $(b\,a)(a\,b)$ in $\alpha(C)$, which again implies that $\alpha(C)$ is not cyclically reduced. In any of the two cases $C$ is not $\alpha$-tight.
\end{proof}

We even can adapt Wrochna's approach for finding $\alpha$-tight walks to our case: We create a digraph $D$ on at most $2|E|$ vertices, such that there is a one-to-one correspondence between cycles of $D$ and $\alpha$-tight cycles of $G$. It then suffices to apply breadth-first-search to detect an $\alpha$-tight cycle of $G$.

\begin{lemma} 
    \label{lem:findtight} 
    There is an algorithm that, given $G$, $H$, and $\alpha : G \to H$, finds in time $O(|V(G)|\cdot |E(G)|)$ an $\alpha$-tight walk or concludes correctly that there is none.
\end{lemma}
\begin{proof} 
    Let $D$ be the digraph whose vertex set is the set of oriented edges $u \rightarrow v$ of $G$ such that $\alpha(u) \neq \alpha(v)$. Create an arc from $u \rightarrow v$ to $u' \rightarrow v'$ in $D$ if $v = u'$ and $\alpha(\{u, v\}) \neq \alpha(\{u', v'\})$.\\

    Then any directed cycle in $D$ corresponds to an $\alpha$-tight cycle in $G$ and vice versa. The digraph $D$ has at most $2 |E(G)|$ vertices and at most $\displaystyle \sum_{v \in V(G)} 2 \binom{\deg(v)}{2} \leq |V(G)| \sum_{v \in V(G)} \text{deg}(v) = O(|V(G)| \cdot |E(G)|)$ arcs. So we can find a directed cycle in $D$ in time $O(|V(G)| \cdot |E(G)|)$ by breadth-first-search.
\end{proof}

\subsubsection{Characterization of realizable walks}
\label{sec:R:realizable}

Analogously to Theorem~\ref{thm:conswro} from~\cite{Wrochna:20}, we obtain the following characterization of realizable walks and a polynomial-time algorithm that constructs from such a walk a corresponding $H$-recoloring sequence.

\begin{theorem} 
    \label{thm:consref}
    Let $\alpha,\beta$ be two $H$-colorings of $G$. Furthermore, let $q \in V(G)$ be any vertex and let $Q$ be a reduced walk from $\alpha(q)$ to $\beta(q)$. Then $Q$ is realizable for $\alpha,\beta,q$ if and only if
    \begin{enumerate}
        \item $Q$ is topologically valid for $\alpha,\beta,q$. \label{itm:topval}
        \item for every $\alpha$-tight closed walk in $G$ and any vertex $v$ on this walk, for any walk $W$ from $v$ to $q$, $Q = \alpha(W)^{-1} \beta(W)$ in $\pi(H)$. \label{itm:tightwalk}
    \end{enumerate}
    Furthermore, there is an algorithm that, given a reduced walk $Q$, constructs a $H$-recoloring sequence $S$ from $\alpha$ to $\beta$ or certifies that $Q$ cannot satisfy one of the previous conditions. This algorithm runs in time $O(|V(G)|^2 + |V(G)| \cdot |Q|)$. The $H$-recoloring sequence $S$ is such that $S(q) = Q$ and for each $v \in V(G)$ the vertex walk $S(v)$ is reduced.
\end{theorem}

\begin{proof} 
    By Lemma~\ref{lem:fundtri} we have that if $Q$ is realizable then~\ref{itm:topval} holds.
    The other statement follows from lemmas~\ref{lem:fundtri} and~\ref{lem:cons}. 
    
    It remains to prove the ``only if''direction.
    Let $Q$ be a reduced walk satisfying~\ref{itm:topval} and~\ref{itm:tightwalk}. For every vertex $v \in V(G)$, let $W$ be a walk from $q$ to $v$ and let $S_v$ be the reduction of $\alpha(W)^{-1}\cdot Q\cdot \beta(W)$. By Lemma~\ref{lem:equigen}, this definition does not depend on the choice of $W$. We claim that the walks $S_v$ record exactly all steps of an $H$-recoloring sequence $S$ from $\alpha$ to $\beta$. Let write $S_v^j$ for the $j$th vertex of $S_v$.

    For any two adjacent vertices $u,v \in V(G)$ such that the walks $S_u$ and $S_v$ are both non empty, let
    \[
        S_u = (a_0\,a_1)(a_1\,a_2)\ldots(a_{n-1}\,a_n),\text{ and } S_v = (b_0\,b_1)(b_1\,b_2)\ldots(b_{n-1}\,b_n)\enspace.
    \]
    Let $W$ be any walk from $q$ to $u$. Then $W\, \cdot \,(u\,v)$ is a walk from $q$ to $v$ and hence,
    \begin{equation}
        S_v = \alpha(W\,(u\,v))^{-1} \cdot Q \cdot\beta(W\,(u\,v)) = (\alpha(v)\,\alpha(u))\cdot S_u\cdot (\beta(u)\,\beta(v))\enspace \text{ in } \pi(H). \label{eq:uv}
    \end{equation}
    
    We define an arc-set $A$ on $V(G)$ that allows us to determine in which order the vertices of $G$ can be recolored. Let $A$ be defined by applying the following four rules for any pair $(u, v) \in V(G) \times V(G)$:
    \begin{itemize}
        \item[Case 1:] $a_1 = b_0$ and $b_1 = a_0$. Since $S_u$ and $S_v$ are reduced we obtain from~\eqref{eq:uv} that $S_u = (a_0 \, b_0)$ and $S_v = (b_0 \, a_0)$, so $u$ and $v$ swap their colors. The two corresponding color changes can be carried out in any order, so we add no arc to $A$ in this case.
        \item[Case 2:] $a_1 = b_0$ and $b_1 \neq a_0$. We must recolor $u$ before $v$. We add the arc $u \to v$ to $A$.
        \item[Case 3:] $a_1 \neq b_0$ and $b_1 = a_0$. We must recolor $v$ before $u$. We add the arc $v \to u$ to $A$.
        \item[Case 4:] $a_1 \neq b_0$ and $b_1 \neq a_0$. By~\ref{eq:uv}, $(\alpha(v)\,\alpha(u)) = \varepsilon$ so $a_0 = b_0$ and $a_1 = b_1$. There is no constraint on which of $u$ or $v$ we have to recolor first, so we add no arc to $A$ in this case.
    \end{itemize}

    \begin{claim}
        The graph $D := (V(G), A)$ is acyclic.
    \end{claim}
    \begin{proof}
         Assume for a contradiction that $D$ contains a cycle $v_1 \to v_2 \to \ldots \to v_n = v_1$. But then the closed walk $C := (v_1\,v_2)\ldots(v_{n-1}\,v_1)$ is $\alpha$-tight. Indeed:
        \begin{itemize}
            \item Let $1 \leq i \leq n-1$. Since $(v_i, v_{i+1}) \in A$, we have in particular that $\alpha(v_i) \neq \alpha(v_{i+1})$.
            \item Let $v_{n+1} := v_2$ and $ 1 \leq i \leq n-1$. Then $v_i \to v_{i+1} \to v_{i+2}$ implies that
            \[
                \alpha(v_{i+2}) = S_{v_{i+2}}^0 = S_{v_{i+1}}^1 = S_{v_i}^2 \neq S_{v_i}^0 = \alpha(v_i)\enspace.
            \]
        \end{itemize}

        Hence, by Item~\ref{itm:tightwalk}, we have $Q = \alpha(W)^{-1}\cdot\beta(W)$ in $\pi(H)$ for any walk $W$ from $v_i$ to $q$. Thus $S_{v_i} = \alpha(W)\cdot Q \cdot\beta(W)^{-1} = \alpha(W)\cdot \alpha(W)^{-1}\cdot\beta(W)\cdot\beta(W)^{-1} = \varepsilon$. But by the construction of $A$, no vertex $w$ such that $S(w) = \varepsilon$ is incident to an arc in $A$; contradiction.
    \end{proof}

    By the previous claim, the graph $(V(G), A)$ is acyclic, so we may consider its vertices in topological order $v_1 < v_2 < \ldots < v_{|V(G)|}$ on $V(G)$. We obtain a $H$-recoloring sequence from $\alpha$ to $\beta$ that satisfies the \pop property by repeating the following until each vertex has reached its target color. For each $i$ from 1 to $|V(G)|$, if the current color of $v_i$ is not $\beta(v_i)$ then recolor $v_i$ from its current color, say $S_{v_i}^j$, to its next color $S_{v_1}^{j+1}$. The \pop property of the resulting $H$-recoloring sequence follows from the topological ordering of the vertices of $V(G)$.

    To conclude the proof we bound the runtime of the construction of the $H$-recoloring sequence.
    The algorithm first chooses walks to construct $S_v$ for each $v \in V(G)$ by breadth-first search in time $O(|E(G)|)$; this also guarantees that $|S_v| \leq 2|V(G)| + |Q|$.
    Then for each edge $uv \in E(G)$, we check if $A$ contains the arc $u \to v$ or $v \to u$ (or neither), and obtain the topological ordering of $V(G)$ in time $O(|E(G)|)$. The corresponding $H$-recoloring sequence can be computed in time $O(|V(G)| \cdot (|V(G)| + |Q|))$. Observe that this algorithm allows us to decide in time $O(|E(G)| \cdot (|V(G)| + |Q|))$ if a walk $Q$ is $H$-realizable by checking the consecutive colors on each edge. If we want to output the entire $H$-coloring at each step, this increases the size of the output by a factor $|V(G)|$, so this is done in $O(|V(G)|^2 \cdot (|V(G)| + |Q|)$.
\end{proof}

Figures~\ref{fig:carthm} and~\ref{fig:carthm2} illustrate different cases when applying Theorem~\ref{thm:consref}.

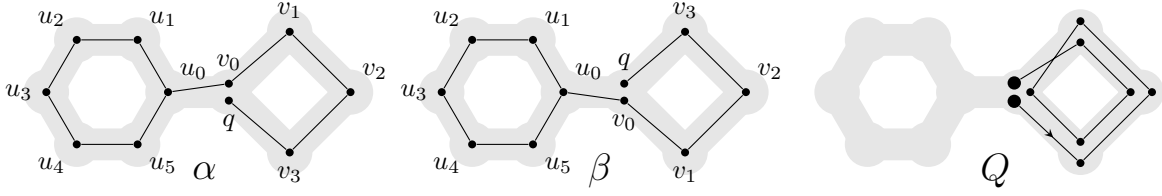
\begin{figure}
\begin{tikzpicture}[scale=0.8]
    \begin{scope}
    
        \dhexasqgraph;
        \node[G,label=60:$u_0$] (v0) at ($(ha0)$){};
        \node[G,label=60:$u_1$] (v1) at ($(ha1)$){};
        \node[G,label=120:$u_2$] (v2) at ($(ha2)$){};
        \node[G,label=180:$u_3$] (v3) at ($(ha3)$){};
        \node[G,label=-120:$u_4$] (v4) at ($(ha4)$){};
        \node[G,label=-60:$u_5$] (v5) at ($(ha5)$){};
        \node[G,label=90:$v_0$] (v6) at ($(hb0) + (90:4pt)$){};
        \node[G,label=90:$v_1$] (v7) at ($(hb1)$){};
        \node[G,label=60:$v_2$] (v8) at ($(hb2)$){};
        \node[G,label=-90:$v_3$] (v9) at ($(hb3)$){};
        \node[G,label=-90:$q$] (v10) at ($(hb0) - (90:4pt)$){};
        \draw[Ge] (v0)--(v1)--(v2)--(v3)--(v4)--(v5)--(v0)--(v6)--(v7)--(v8)--(v9)--(v10);
        \node at (1.6,-1.3) {\Large$\alpha$};
        
    \end{scope}
    \begin{scope}[shift={(6.5,0)}]
        \dhexasqgraph;
        \node[G,label=60:$u_0$] (v0) at ($(ha0)$){};
        \node[G,label=60:$u_1$] (v1) at ($(ha1)$){};
        \node[G,label=120:$u_2$] (v2) at ($(ha2)$){};
        \node[G,label=180:$u_3$] (v3) at ($(ha3)$){};
        \node[G,label=-120:$u_4$] (v4) at ($(ha4)$){};
        \node[G,label=-60:$u_5$] (v5) at ($(ha5)$){};
        \node[G,label=-90:$v_0$] (v6) at ($(hb0) - (90:4pt)$){};
        \node[G,label=-90:$v_1$] (v7) at ($(hb3)$){};
        \node[G,label=60:$v_2$] (v8) at ($(hb2)$){};
        \node[G,label=90:$v_3$] (v9) at ($(hb1)$){};
        \node[G,label=90:$q$] (v10) at ($(hb0) + (90:4pt)$){};
        \draw[Ge] (v0)--(v1)--(v2)--(v3)--(v4)--(v5)--(v0)--(v6)--(v7)--(v8)--(v9)--(v10);
        \node at (1.6,-1.3) {\Large$\beta$};
    \end{scope}
    
    \begin{scope}[shift={(13,0)}]
        \dhexasqgraph;
        \node[G1] (v0) at ($(hb0) + (-120:5pt)$){};
        \node[G] (v1) at ($(hb3) + (-90:5pt)$){};
        \node[G] (v2) at ($(hb2) + (0:5pt)$){};
        \node[G] (v3) at ($(hb1) + (90:5pt)$){};
        \node[G] (v4) at ($(hb0) + (0:5pt)$){};
        \node[G] (v5) at ($(hb3) + (90:5pt)$){};
        \node[G] (v6) at ($(hb2) - (0:5pt)$){};
        \node[G] (v7) at ($(hb1) + (-90:5pt)$){};
        \node[G1] (v8) at ($(hb0) + (120:5pt)$){};
        \draw[Ge] (v0)--(v1) node[sloped,pos=0.55,allow upside down]{\arrowInG}; ;
        \draw[Ge] (v1)--(v2)--(v3)--(v4)--(v5)--(v6)--(v7)--(v8);
        \node at (1.6,-1.3) {\Large$Q$};
        
    \end{scope}
\end{tikzpicture}

\caption{The cycle $C = (u_0 \, u_1)(u_1 \, u_2) (u_2 \, u_3) (u_3 \, u_4) (u_4 \, u_5) (u_5 \, u_0)$ is $\alpha$-tight, so its vertex are frozen, so letting $W = (u_0 \, v_0) (v_0 \, v_1)(v_1 \, v_2) (v_2 \, v_3) (v_3 \, v_4) (v_4 \, q)$, we obtain that $Q = \alpha(W)^{-1} \beta(W)$ is the only possibly $H$-realizable walk for $\alpha,\beta,q$. Then we can check that it is $H$-realizable.}

\label{fig:carthm}
\end{figure}

\begin{figure}
\begin{tikzpicture}[scale=0.8]
    \begin{scope}
        \dhexasqgraph;
        \node[G,label=60:$u_0$] (v0) at ($(ha0)$){};
        \node[G,label=60:$u_1$] (v1) at ($(ha1)$){};
        \node[G,label=120:$u_2$] (v2) at ($(ha2)$){};
        \node[G,label=180:$u_3$] (v3) at ($(ha3)$){};
        \node[G,label=-120:$u_4$] (v4) at ($(ha4)$){};
        \node[G,label=-60:$u_5$] (v5) at ($(ha5)$){};
        \node[G,label=90:$v_0$] (v6) at ($(hb0) + (90:4pt)$){};
        \node[G,label=90:$v_1$] (v7) at ($(hb1)$){};
        \node[G,label=60:$v_2$] (v8) at ($(hb2) + (90:4pt)$){};
        \node[G,label=-60:$v_3$] (v8bis) at ($(hb2) - (90:4pt)$){};
        \node[G,label=-90:$v_4$] (v9) at ($(hb3)$){};
        \node[G,label=-90:$q$] (v10) at ($(hb0) - (90:4pt)$){};
        \draw[Ge] (v0)--(v1)--(v2)--(v3)--(v4)--(v5)--(v0)--(v6)--(v7)--(v8)--(v8bis)--(v9)--(v10)--(v6);
        \node at (1.6,-1.3) {\Large$\alpha$};
        
    \end{scope}
    \begin{scope}[shift={(6.5,0)}]
        \dhexasqgraph;
        \node[G,label=60:$u_0$] (v0) at ($(ha0)$){};
        \node[G,label=60:$u_1$] (v1) at ($(ha1)$){};
        \node[G,label=120:$u_2$] (v2) at ($(ha2)$){};
        \node[G,label=180:$u_3$] (v3) at ($(ha3)$){};
        \node[G,label=-120:$u_4$] (v4) at ($(ha4)$){};
        \node[G,label=-60:$u_5$] (v5) at ($(ha5)$){};
        \node[G,label=-90:$v_0$] (v6) at ($(hb0) - (90:4pt)$){};
        \node[G,label=-90:$v_1$] (v7) at ($(hb3)$){};
        \node[G,label=-60:$v_2$] (v8) at ($(hb2) - (90:4pt)$){};
        \node[G,label=60:$v_3$] (v8bis) at ($(hb2) + (90:4pt)$){};
        \node[G,label=90:$v_4$] (v9) at ($(hb1)$){};
        \node[G,label=90:$q$] (v10) at ($(hb0) + (90:4pt)$){};
        \draw[Ge] (v0)--(v1)--(v2)--(v3)--(v4)--(v5)--(v0)--(v6)--(v7)--(v8)--(v8bis)--(v9)--(v10)--(v6);
        \node at (1.6,-1.3) {\Large$\beta$};
    \end{scope}
    \begin{scope}[shift={(13,0)}]
        \dhexasqgraph;
        \node[G1] (v0) at ($(hb0) + (-120:5pt)$){};
        \node[G] (v1) at ($(hb3) + (-90:5pt)$){};
        \node[G] (v2) at ($(hb2) + (0:5pt)$){};
        \node[G] (v3) at ($(hb1) + (90:5pt)$){};
        \node[G] (v4) at ($(hb0) + (0:5pt)$){};
        \node[G] (v5) at ($(hb3) + (90:5pt)$){};
        \node[G] (v6) at ($(hb2) - (0:5pt)$){};
        \node[G] (v7) at ($(hb1) + (-90:5pt)$){};
        \node[G1] (v8) at ($(hb0) + (120:5pt)$){};
        \draw[Ge] (v0)--(v1) node[sloped,pos=0.55,allow upside down]{\arrowInG}; ;
        \draw[Ge] (v1)--(v2)--(v3)--(v4)--(v5)--(v6)--(v7)--(v8);
        \node at (1.6,-1.3) {\Large$Q$};
    \end{scope}
\end{tikzpicture}

\caption{The same cycle is $\alpha$-tight, so setting the same walk $W = (u_0 \, v_0) (v_0 \, v_1)(v_1 \, v_2) (v_2 \, v_3) (v_3 \, v_4) (v_4 \, q)$, $Q := \alpha(W)^{-1} \beta(W)$ is the only possibly $H$-realizable walk. Here however, $Q$ is not topologically valid since letting $C := (q \, v_0) (v_0 \, v_1)(v_1 \, v_2) (v_2 \, v_3) (v_3 \, v_4) (v_4 \, q)$, we have that $Q = \alpha(C)^{-2}$ in $\pi(H)$ so $Q^{-1} \alpha(C) Q = \alpha(C) \neq \alpha(C)^{-1} = \beta(C)$ in $\pi(H)$. 
}
\label{fig:carthm2}
\end{figure}
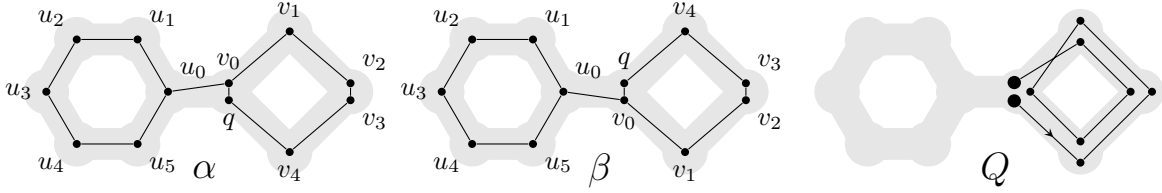

\subsubsection{Proof of Theorem~\ref{thm:R:realwalks}}

\begin{proof}
Recall that $\undir{\Pi}$ is the set of all $H$-realizable walks for $\alpha,\beta,q$.

We give an algorithm that decides which case of Theorem~\ref{thm:R:realwalks} applies.
First, apply Lemma~\ref{lem:findtight}. If it finds an $\alpha$-tight walk, then by Theorem~\ref{thm:consref} there is only one possibly $H$-realizable walk $Q = \alpha(W)^{-1} \beta(W)$ where $W$ is any walk from a vertex of that cycle to $q$; and we can check whether $Q$ is $H$-realizable, so we conclude with $\undir{\Pi} = \emptyset$ or $\undir{\Pi} = Q$  in time $O(|V(G)| \cdot (|V(G)| + |Q|)) = O(|V(G)|\cdot|E(G)|)$ since $|Q| \leq |E(G)|)$. If there is no $\alpha$-tight walk then, by Theorem~\ref{thm:consref}, we have $\undir{\Pi} = \Pi$ where $\Pi$ is the set of topologically valid walks returned by Theorem~\ref{thm:T:Wrochna}.

Reading the classification of Theorem~\ref{thm:T:Wrochna}, we immediately see that case~\ref{itm:R:allwalks} of Theorem~\ref{thm:R:realwalks} occurs when $\alpha(C) = \beta(C) = \varepsilon$ in $\pi(H)$ for all closed walks $C$ in $G$.
\end{proof}

\subsection{Reflexive directed graphs}
\label{sec:R:directed}

For the remainder of this section
let $H$ be a fixed reflexive digraph. We also fix any reflexive instance of \recol{H} by choosing a reflexive digraph $G$. Notice that if $H$ contains no triangle of algebraic girth 1 then any $H$-recoloring sequence satisfies the \pop property: Suppose that in a step of a $H$-recoloring sequence a vertex $u$ changes its color from $a$ to $b$ and let $v$ be a neighbor of $u$. 
If the color of $v$ is not $a$ or $b$ during the change, then $H$ contains the triangle of algebraic girth 1 shown in Figure~\ref{fig:onegirth}. 

Given two $H$-colorings $\alpha,\beta \colon G \to H$ and $q \in V(G)$, we say that a walk $Q$ from $\alpha(q)$ to $\beta(q)$ is \emph{$H$-realizable} if there is a $H$-recoloring sequence $S$ satisfying the \pop property such that $Q = S(q)$. We observed that any  $H$-realizable walk $Q$ is also $\undir{H}$-realizable. Hence, we may apply Theorem~\ref{thm:R:realwalks} to determine \emph{all} $\undir{H}$-realizable walks and then check if one of them is $H$-realizable.
That is, for each of the four classes of $\undir{H}$-realizable walks in  Theorem~\ref{thm:R:realwalks}, we give a polynomial-time algorithm that finds a $H$-realizable walk in the class or indicates correctly that no such walk exists. Clearly, in Case~\ref{itm:R:nowalk} of Theorem~\ref{thm:R:realwalks} there is no $H$-realizable walk (since there is no $\undir{H}$-realizable walk). It remains to treat cases~\ref{itm:R:onlyQ} -- \ref{itm:R:allwalks} of Theorem~\ref{thm:R:realwalks}, which we will do separately in the next subsections.

\subsubsection{Case~\ref{itm:R:onlyQ} of Theorem~\ref{thm:R:realwalks}}

Let $\alpha, \beta : G \to H$ and $q \in V(G)$. We suppose that we are in Case
Case~\ref{itm:R:onlyQ} of Theorem~\ref{thm:R:realwalks}. That is, the set $\undir{\Pi}$ of \undir{H}-realizable walks for $\alpha, \beta, q$ contains only a single walk $Q$ (from $\alpha(q)$ to $\beta(q)$).

\begin{lemma}[The move forward algorithm]
    \label{lem:moveforward} 
    There is a polynomial-time algorithm that, given $\alpha$, $\beta$, $q$, and $Q$, determines if the walk $Q$ is $H$-realizable. If so, it constructs an associated $H$-recoloring sequence of polynomial length.
\end{lemma}
\begin{proof} 
    The algorithm proceeds as follows.
    Construct all walks $S_v$ from $\alpha(v)$ to $\beta(v)$ by choosing walks from $q$ to $v$ (if $Q$ is realizable, then it can be associated with a $H$-recoloring sequence whose vertex walks are the walks $S_v$). For any vertex $u \in V(G)$ and its vertex walk $S_u = (a_0\,a_1)\ldots(a_{n-1}\,a_n)$, we call \emph{moving forward} $u$ the operation that changes the color of $u$ to its next color in $S_u$ (if the color of $u$ is $a_i$, then moving forward $u$ changes its color to $a_{i+1}$). Let $v$ be a neighbor of $u$ and suppose that the color of $u$ is $a_i$ ($i<n$). We write $u < v$ if:
    \begin{itemize}
        \item the color of $v$ is $a_{i+1}$ (think that $v$ is one step forward in their common walk so $u$ needs to be pulled to the color of $v$ first), or
        \item the color of $v$ is $a_i$ and $u \rightarrow v \in A(G)$ and $a_i \rightarrow a_{i+1} \not\in A(H)$, or
        \item the color of $v$ is $a_i$ and $v \rightarrow u \in A(G)$ and $a_{i+1} \rightarrow a_i \not\in A(H)$ (think that $u$ and $v$ are at the same step in their walk, but the orientation of $H$ prevents $v$ from being pushed, so we must push $u$ first)
    \end{itemize}

    So $u<v$ means that $u$ must move forward before $v$ does. We say that $u$ \emph{is able to move forward} if there is no neighbor $v$ of $u$ such that $v<u$ and if $u$ is not at the end of $S_u$. Observe that moving $u$ forward (when possible) can only enable its its neighbors to move forward. So we can deduce an algorithm to decide whether $Q$ is realizable:
    
    \begin{enumerate}
        \item \label{itm:mfw:init} Initialization: For all arc $u \rightarrow v \in A(G)$, decide if gives an order $u<v$, $v>u$, or none. For all vertex $u \in V(G)$, let $b(u)$ be the number of neighbors $v$ of $u$ such that $v<u$. Also set a list $M$ of all vertices that are able to move forward, so of vertices $u \in V(G)$ such that $S_u$ is non empty and $b(u) =0$.
        \item \label{itm:mfw:push} While $M$ is not empty, pick any vertex $u \in M$ and move $u$ forward. Then update $S_u$ by deleting its first edge. check the changes of all arcs $u \rightarrow v$ and $v \rightarrow u$ to update the number of constraints $b(v)$ (which might decrease by $1$) and $b(u)$. If $b(v) = 0$ and $S_v$ is non empty, then add $v$ to $M$. Similarly if $b(u)$ is still $0$ and $S_u$ is non empty, then add $u$ to $M$.
    \end{enumerate}

    This algorithm can stop for two reasons: either the last $H$-coloring of $G$ is $\beta$, so the moves performed by the algorithm yield a $H$-recoloring sequence from $\alpha$ to $\beta$ and $Q$ is $H$-realizable. Or there is a cycle $u_1 < u_2 < \ldots u_n < u_1$ (it is possible that $n =2$ if there is a symmetric edge $ u_1 \rightleftarrows u_2$ and $S_{u_1}$ or $S_{u_2}$ is not symmetric). This means that $Q$ is not $H$-realizable and the obstruction is the cycle $(u_1u_2)\ldots(u_nu_1)$.
    
     Phase~\ref{itm:mfw:init} takes time $O(|E(G)|)$. Phase~\ref{itm:mfw:push} will require at most $|Q| + 2|V(G)|$ move forward for each vertex $v \in V(G)$, so it runs in time $O(\sum_{v \in V(G)} (|Q| + |V(G)|) \cdot \text{deg}(v)) = O(|E(G)| \cdot (|Q| + |V(G)|)$, which is thus the final complexity of the algorithm.
\end{proof}

\subsubsection{Case \ref{itm:R:walksRnP} of Theorem~\ref{thm:R:realwalks}}

Let $\alpha, \beta : G \to H$ and $q \in V(G)$. We assume in this section that we are in Case
Case~\ref{itm:R:walksRnP} of Theorem~\ref{thm:R:realwalks}. That is, the set $\undir{\Pi}$ of \undir{H}-realizable walks for $\alpha, \beta, q$ is the set $\{ R^nP \mid n \in \mathbb{Z} \}$, where $R, P\in \pi(H)$ and $R$ is a closed walk. We show that in order to decide whether there is a walk of the form $R^nP$ that is $H$-realizable, it suffices to check the walks $R^iP$ for all $-N \leq i \leq N$, where $N$ is polynomial in $|V(G)|$ and $|V(H)|$.

\begin{lemma} 
	\label{lem:maxred} 
	Let  $R = A R_0 A^{-1}$, such that so $A$ is reduced and $R_0$ is cyclically reduced.
	Then there exists some $n_0 \in \mathbb{N}$,
	polynomial in $|V(G)|$ and $|V(H)|$, such that for $n \geq n_0$, reducing the walk 
	\[ 
		S_v = \alpha(W)^{-1}\cdot A \cdot R_0^{n} \cdot A^{-1} \cdot P \cdot \beta(W) 
	\]
	leaves the middle term $R_0^{n-n_0}$ intact.
\end{lemma}

\begin{proof} 
	Let $a_0:= \lceil |V(G)|/|R_0| \rceil$ and $b_0:= \lceil(|P| +
	|V(G)|)/|R_0| \rceil$. Let $n_0 := a_0 + b_0$. Let
	$W$ be a walk from $q$ to $v$ of length at most $V(G)$. Such a walk can be found by breadth-first-search in time $O(E(G))$. Let $n > n_0$ and consider the walk 
	\[
    	S_v = \alpha(W)^{-1} \cdot R^n\cdot P \cdot \beta(W) 
    \]
	Since $R_0$ is cyclically reduced, we have that $R^n$ reduces to $A \cdot R_0^n\cdot  A^{-1}$. Furthermore, since $n > n_0$ and $|\alpha(W)^{-1}| \leq |W| \leq |V(G)| \leq a_0|R_0| = |R_0^{a_0}|$, the
	left term $\alpha(W)^{-1}$ reduces at most with $A \cdot R_0^{a_0}$. Similarly, the
	right term $P \cdot\beta(W)$ at most reduces with $R_0^{b_0} \cdot A^{-1}$.
	Finally, the middle term $R_0^{n-n_0}$ cannot reduce in $S_v$.
\end{proof}

\begin{lemma} 
    \label{lem:caseRnP}
    There exists $N \in \mathbb{N}$ polynomial in $|V(G)|$ and $|V(H)|$ such that the following are equivalent:
    \begin{itemize}
        \item None of the walks $R^nP$ is $H$-realizable for $-N \leq n \leq N$.
        \item None of the walks $R^nP$ is $H$-realizable for $n \in \mathbb{Z}$.
    \end{itemize}
    In particular, we can solve \recolun{H} in polynomial time using the move forward algorithm of Lemma~\ref{lem:moveforward}.
\end{lemma}

\begin{proof} 
    Decompose $R = AR_0A^{-1}$, with all walks reduced and $R_0$ cyclically reduced. Apply Lemma~\ref{lem:maxred} to obtain $n_0$ such that for $n \geq n_0$, in all generated walks
    \[ 
        S_v = \alpha(W)^{-1} \cdot A \cdot R_0^{n} \cdot A^{-1} \cdot P \cdot \beta(W)\enspace,
    \]
    where $W$ is a walk from $q$ to $v$ of length at most $V(G)$ and
    a middle term $R_0^{n-n_0}$ does not reduce. 
    Let $a_0 := \lceil |V(G)|/|R_0| \rceil$. Let $N := n_0 + 2 a_0$. We assume that $R^{N}P$ is not $H$-realizable and apply the move forward algorithm of Lemma~\ref{lem:moveforward} to $R^{N}P$ to get an obstruction cycle $ C =u_1 < u_2 < \ldots <u_r <u_1$ (recall that possibly $r=2$). We can wonder at what step of their respective walk the vertices of $C$ are blocked by the obstruction.

    For $v \in V(G)$ we decompose $S_v$ into three parts: $\alpha(W)^{-1} A R_0^{a_0}$ is the first part. $R_0^{n-n_0}$ is the middle part and $R_0^{a_0} \cdot A^{-1} \cdot P \cdot \beta(W)$ the last part. Since $|C| \leq |V(G)| \leq |R_0^{a_0}|$ we have that there cannot be two vertices $u,v \in V(C)$ such that $u$ is blocked before a move of the first part of $S_u$ and $v$ before a move of the last part of $S_v$. 

\begin{itemize}
\item If all vertices are blocked before a move in the first or in the middle part of their vertex walk, then the move forward algorithm will stop at the same steps as well for any $n \geq N$ (since all walks are the same to that step).
\item If all vertices stop before a move contained in the middle or in the last part of their vertex walks, then this means there is an obstruction in the first or in the last part of the reverse walks $S_v^{-1}$. So as before, using the move forward algorithm to test if that reverse walk reconfigures $\beta$ to $\alpha$, it will stop at the exact same step for $n \geq N$. 
\end{itemize}

In all cases, the same cycle $u_1 < \ldots < u_r <u_1$ will be an obstruction to the move forward algorithm for $n \geq N$, so none of the walks $R^nP$ is $H$-realizable for $n \geq N$.
Exchanging $R_0$ with $R_0^{-1}$, we obtain that the same holds for $n \leq -N$ if $R^{-N}$ is not $H$-realizable.
\end{proof}

The following lemma summarizes the previous results.

\begin{lemma} 
    \label{lem:R:case3} 
    If Case~\ref{itm:R:walksRnP} of Theorem~\ref{thm:R:realwalks} applies then there is a polynomial-time algorithm that finds an $H$-realizable walk or concludes correctly there is none.
\end{lemma}

\begin{proof} 
    Determine $N$ according to~\ref{lem:caseRnP}. Then, for each $-N \leq i \leq N$, we can treat the walk $R^iP$ as in Case~\ref{itm:R:onlyQ} of Theorem~\ref{thm:R:realwalks}: we apply the move-forward-algorithm of Lemma~\ref{lem:moveforward} to $R^iP$. If one of $R^iP$ is $H$-realizable then we are done, otherwise we invoke Lemma~\ref{lem:caseRnP} and conclude that none of the walks $R^iP$ is $H$-realizable for any $i \in \mathbb{Z}$.
\end{proof}

\begin{remark} 
    With Lemma~\ref{lem:R:case3}, Case~\ref{itm:R:walksRnP} is treated in polynomial time. Since we apply the move forward algorithm $O(|N|) = O(|P|+|V(G)|) = O(|V(G)| \cdot |E(G)| + |E(H)|)$ times, the total running time is $O((|V(G)| \cdot (|V(G)| \cdot |E(G)| + |E(H)|))^2)$.
\end{remark}

\subsubsection{Case~\ref{itm:R:allwalks} of Theorem~\ref{thm:R:realwalks}}

Let $\alpha, \beta : G \to H$ and $q \in V(G)$. We assume in this section that we are in Case
Case~\ref{itm:R:allwalks} of Theorem~\ref{thm:R:realwalks}. 
So for any closed walk $C$ of $G$, we have $\alpha(C) = \beta(C) = \varepsilon$ in $\pi(H)$.
Our goal is to decide if some reduced walk from $\alpha(q)$ to $\beta(q)$ is $H$-realizable.

\begin{lemma} 
    \label{lem:R:dirsym} 
    Let $Q$ be a reduced walk from $\alpha(q)$ to $\beta(q)$ (so $Q$ is $\undir{H}$-realizable). For any vertex $v \in V(G)$, let $S_v$ be the reduced walk $\alpha(W)^{-1}Q\beta(W)$, where $W$ is any walk from $q$ to $v$. Then $Q$ is $H$-realizable if and only if for any directed closed walk $u_1 \rightarrow u_2 \rightarrow \ldots \rightarrow u_n \rightarrow u_1$, for each $1 \leq k \leq n$, the walk $S_{u_k}$ is symmetric ($n \geq 2$, possibly $n=2$ and the closed walk goes forth and back a symmetric edge).
\end{lemma}

\begin{proof} Assume that $Q$ is $H$-realizable and let $C = u_1 \rightarrow u_2 \rightarrow \ldots \rightarrow u_n \rightarrow u_1$ be a directed cycle. Let $u_{n+1} := u_1$ and $u_0 := u_n$. Let $1 \leq k \leq n$ and consider any color change of $u_k$ in the $H$-recoloring sequence generated by $Q$, say $(a_ia_{i+1})$. Let $\delta$ be the $H$-coloring just before this step.

\begin{itemize}
\item If both $u_{k-1}$ and $u_{k+1}$ have color $a_{i+1}$, then the arcs $u_{k-1} \rightarrow u_k \rightarrow u_{k+1}$ yield $a_i \rightleftarrows a_{i+1}$, because $\delta$ is a homomorphism.
\item If both $u_{k-1}$ and $u_{k+1}$ have color $a_i$, then the same holds because we still have a homomorphism after the color change.
\item If $u_{k-1}$ has color $a_i$ and $u_{k+1}$ has color $a_{i+1}$ then we immediately have $a_i \rightarrow a_{i+1}$. Considering $\delta(C) = \alpha(C) = \varepsilon$ in $\pi(H)$, we obtain that $\delta((u_{k+1} \rightarrow \ldots \rightarrow u_{k-1})$ is a path in $H$ that contains the inverse arc $a_{i+1} \rightarrow a_i$. In particular, $a_i \rightleftarrows a_{i+1}$.
\item Similarly, if $u_{k-1}$ has color $a_{i+1}$ and $u_{k+1}$ has color $a_i$, then $a_i \rightleftarrows a_{i+1}$.
\end{itemize}

In any case, $a_i \rightleftarrows a_{i+1}$. So $S_{u_k}$ is symmetric. This concludes the only if part.\\

\emph{Conversely,} if $Q$ is not $H$-realizable, then using the move forward algorithm of lemma~\ref{lem:moveforward}, there is an obstruction closed walk $C = (u_1u_2)\ldots(u_nu_1)$ in $G$ and let $\delta \colon G \to H$ be the coloring obtained at the end of the algorithm, so we have obstructions $u_1 < u_2 < \ldots < u_n < u_1$. As $\delta(C) = \alpha(C) = \varepsilon$ in $\pi(H)$, we can write $\delta(C) = PP^{-1}$ in $H$ (as walks in $H$, without reduction).\\

If $P$ contains at least two vertices, then let $u \in V(C)$ such that $\delta(v)$ is an extremity of $P$. Say $u := u_k$ for $1 \leq k \leq n$. Let $u_p,u_q$ be the nearest vertices of $u_k$ in both directions of $C$ such that $\delta(u_p) \neq \delta(u_k) \neq \delta(u_q)$. So we have 
\begin{center}
$\delta(u_p) \neq \delta(u_{p+1}) = \ldots = \delta(u_k) = \ldots \delta(u_{q-1}) \neq \delta(u_q)$. 
\end{center}
Because $\delta(u_k)$ is an extremity of $P$, we can state that $\delta(u_p) = \delta(u_q)$. But then the obstruction $u_p < u_{p+1}$ and $u_{q-1} < u_q$ cannot both hold. This contradicts the assumption that $C$ is an obstruction closed walk.\\

So $P$, and then $\delta(C)$ contains only one vertex $a \in V(H)$. We deduce from the push or pull property that all vertices of $C$ have the same next move forward onto the same vertex $b \in V(H)$. First assume that $u_1 \rightarrow u_2$ in $A(G)$, then $u_1 < u_2$ means that $b \rightarrow a \in A(H)$ and $a \rightarrow b \not \in A(H)$. Considering that $u_k < u_{k+1}$ for all $1 \leq k \leq n$ (with $u_{n+1} = u_1$), we obtain that necessarily $u_k \rightarrow u_{k+1}$ in $A(G)$, so $C$ is directed. Similarly, if $u_1 \leftarrow u_2$, we obtain that $C$ is directed in the other way.\\

In all cases, $C$ is a directed closed walk and the edge $(ab)$ is not symmetric. So the walks $S_{u_k}$ ($1 \leq k \leq n)$ are not symmetric.
\end{proof}

\subsubsection{Main result on directed graphs}

We can now prove the following result from which theorem~\ref{thm:triangledir} follows immediately.

\begin{theorem} Let $H$ be any reflexive graph and $\alpha,\beta \colon G \to H$ be two $H$-coloring where $G$ is reflexive. We can find in polynomial time a $H$-recoloring sequence from $\alpha$ to $\beta$ that satisfies the push or pull property or conclude correctly that there is none.

\end{theorem}

\begin{proof} Pick any vertex $v \in V(G)$ and use theorem~\ref{thm:R:realwalks} to obtain a description of the set of all $\undir{H}$-realizable walks in time $O(|V(G)|\cdot|E(G)| + |E(H)|)$. If there is none, we conclude that there is no $H$-realizable walk. If there is only one $\undir{H}$-realizable $Q$, then we only need to test $Q$ with the move-forward algorithm of Lemma~\ref{lem:moveforward}. In the third case, lemma~\ref{lem:R:case3} concludes.

In case~\ref{itm:R:allwalks}, apply Tarjan's algorithm to obtain in time $O(|E(G)|)$ all strongly connected components of $G$ and hence the set $V'$ of all vertices that belong to some directed closed walk. By lemma~\ref{lem:R:dirsym}, $H$-realizable walks are exactly those which generate symmetric walks on $V'$. If $V' = \emptyset$, then any walk from $\alpha(q)$ to $\beta(q)$ is $H$-realizable, so we can conclude in time $O(|E(G)|)$ using breadth-first search. Otherwise, let $q_0 \in V'$ and use lemma~\ref{lem:symwalks} to determine the set of walks from $\alpha(q_0)$ that generate symmetric walks on $V'$ with any system $(W_v)$ of walks from $q$ to $v$ for all $v \in V'$ (since $\undir{H}$-realizable walks are topologically valid, Lemma~\ref{lem:equigen} ensures that generated walks to not depend of how the walks $W_v$ are constructed).

\begin{itemize}
\item If this set is empty, then there is no $H$-realizable walks for $\alpha,\beta,q_0$, so there is no $H$-recoloring sequence satisfying the push or pull property.
\item If there is only one walk $Q$ that is realizable for $\alpha,\beta,q_0$, then we can construct the wanted $H$-recoloring sequence by using the move forward algorithm of lemma~\ref{lem:moveforward} in polynomial time. 
\item If all walks from $\alpha(q_0)$ to $\beta(q_0)$ generate symmetric walks on $V'$, then we can search for one by breadth-first search in time $O(|E(G)|)$. If there is none, then there is no $H$-recoloring sequence satisfying the push or pull property. If there is some, then the move forward algorithm can construct an associated $H$-recoloring sequence.
\end{itemize}

\end{proof}

\bibliographystyle{plainurl}
\bibliography{references} 

\end{document}